\documentclass[
preprint,aps,
longbibliography,
showkeys
]{revtex4-2}

\usepackage{graphicx}%
\usepackage{amsmath,amssymb,amsfonts,amsthm}%
\usepackage{bm}%
\usepackage{xcolor}%
\usepackage{booktabs}%
\usepackage{algorithm}%
\usepackage{algorithmicx}%
\usepackage{algpseudocode}%
\usepackage{listings}%

\usepackage{xfrac}
\usepackage{comment}
\usepackage{enumitem} 
\usepackage{bigints} 
\usepackage{orcidlink}
\usepackage{cancel}
\usepackage{multirow}
\usepackage{mathrsfs}
\usepackage[section]{placeins} 
\usepackage{dcolumn}

\usepackage{etoolbox} 
\usepackage{subfig} 
\usepackage{caption}

\usepackage{hyperref}
\hypersetup{
    colorlinks=true,
    linkcolor=blue,
    citecolor=blue,
    filecolor=blue,
    urlcolor=blue,
}

\newtheorem{theorem}{Theorem} 
\newtheorem{lemma}[theorem]{Lemma}
\newtheorem{example}{Example}%
\newtheorem{remark}{Remark}%
\newtheorem{definition}{Definition}%

\newcommand{\coupledsum}[3]{\sum_{#1}^{#2}{\!}_{\raisebox{-0.3ex}{\scalebox{1.5}{$\scriptscriptstyle#3$}}}}
\newcommand{\coupledprod}[3]{\prod_{#1}^{#2}{\!}_{\raisebox{-0.3ex}{\scalebox{1.5}{$\scriptscriptstyle#3$}}}}

\usepackage{csquotes}

\begin{document}

\title[Coupled Entropy]{On the uniqueness of the coupled entropy}


\author{{Kenric Nelson} \orcidlink{0000-0001-6962-7459}}\email{kenric.nelson@photrek.io}

\affiliation{Photrek, Inc. \\ 
56 Burnham St. $\#1$\\
Watertown, MA 02472, USA}

\date{\today}

\begin{abstract}
The coupled entropy, $H_\kappa,$ is proven to uniquely satisfy the requirement that a generalized entropy be a measure of the uncertainty at the scale, $\sigma,$ for a class of non-exponential distributions. The coupled stretched exponential distributions, including the generalized Pareto and Student's t distributions, are uniquely parameterized to quantify linear uncertainty with the scale and nonlinear uncertainty with the tail shape for a broad class of complex systems. Thereby, the coupled entropy optimizes the representation of the uncertainty due to linear sources. Lemmas for the composability and extensivity of the coupled entropy are proven. The uniqueness of the coupled entropy is further supported by demonstrating consistent thermodynamic relationships, which correspond to a model used for the momentum of high-energy particle collisions. Applications of the coupled entropy in measuring statistical complexity, training variational inference algorithms, and designing communication channels are reviewed.

\end{abstract}



\maketitle

\section{Introduction}

Accurate quantification of uncertainty \cite{smith_uncertainty_2013} is foundational to analysis, modeling, and design throughout the sciences. Contemporary science and engineering are increasingly dependent on nonlinear systems that increase the complexity of uncertainty measurement. As biologists apply physical chemistry to increasingly complex systems \cite{marino_methodology_2008, eskov_classification_2019}, the fluctuations in uncertainty complicate these measurements. As computer scientists accelerate the mimicking of biological intelligence, the training and modeling of artificial intelligence can no longer depend on assumptions of linearity and exponential tail decay \cite{almeida_predictive_2002,benderskaya_nonlinear_2013}. The two principal methodologies of uncertainty quantification are statistical measurement \cite{walpole_probability_1978} of the random variable, such as the mean, variance, and higher moments, and statistical measurement of the probability distribution in the form of an entropy function \cite{jaynes_probability_2003}. In this paper, I will demonstrate that there is a precise relationship between the uncertainty measurement of scale, which is the standard deviation for the stretched exponential distributions, and the entropy of a distribution, which for the stretched exponentials is equal to the density at the scale. I will prove that this relationship between the density at the scale and the entropy of the distribution specifies a unique requirement for defining a generalization of the entropy for nonlinear systems that is only fulfilled by the coupled entropy \cite{nelson_average_2017}.

Infodynamics \cite{salthe_what_2001} refers to the overlapping methodologies of information theory \cite{shannon_mathematical_1948, jaynes_information_1957}, thermodynamics \cite{fermi_thermodynamics_2012, lewis_thermodynamics_2020}, and statistical mechanics \cite{ma_statistical_1985, davidson_statistical_2013}. The Boltzmann-Gibbs-Shannon (BGS) entropy is well-established as the appropriate uncertainty metric for equilibrium systems with linear sources of uncertainty \cite{goldstein_gibbs_2019}. Given moment constraints, the BGS entropy is maximized by members of the exponential family of distributions, which E.T. Jaynes showed provides a broad principle for model development \cite{jaynes_information_1957-1} .  Uncertainty quantification of non-equilibrium systems with nonlinear sources of uncertainty presents significantly more difficult challenges. First, the moments depend on both the scale and shape of the distributions, and diverge for moments greater than the inverse of the shape parameter. Secondly, the BGS entropy is quickly dominated by the shape of the distribution, breaking the exponential family connection between the scale and entropy. The distributions defined by Pareto \cite{arnold_pareto_2015} and Gosset \cite{pearson_students_1942} are shown to be of critical importance for defining a generalized entropy function, since they each utilize a definition of scale that is independent of the shape of the distribution, which is induced by the nonlinear sources of uncertainty. 

While the Rényi \cite{renyi_measures_1961}; Sharma, Mohan, and Mitter \cite{sharma_measures_1978, sharma_new_1975, nielsen_closed-form_2012}; and Tsallis \cite{tsallis_possible_1988, tsallis_nonadditive_2009-1} entropies, in particular, and many other generalized entropies \cite{hanel_comprehensive_2011, tempesta_beyond_2016, beck_generalised_2009} have shown promise in modeling the non-exponential distributions characteristic of complex adaptive systems, lingering questions about their derivation from first principles \cite{cho_fresh_2002} have restricted their applicability. The goal of this paper is to demonstrate that by fulfilling a unique requirement regarding the solution of a generalized entropy, applications of such a generalization can be placed on a stronger theoretical foundation.

The paper begins with Preliminaries \ref{sec_prelim} defining the Coupled Exponential Family and the generalized functions that provide a concise pseudo-algebra for the expressions. The prototypical process with multiplicative noise is reviewed to motivate the utility of the coupled exponential definitions. In the Main Results section \ref{secResults}, I prove that the coupled entropy provides a unique solution for the measurement of uncertainty for nonlinear systems. The origin of this solution begins with the identification that the source of nonlinearity in a complex system is equal to the non-exponential tail shape; as such, this property is called the nonlinear statistical coupling, $\kappa,$ \cite{nelson_nonlinear_2010, nelson_average_2017}. Second, a proof is provided that there is only one definition of the scale, $\sigma,$ which is independent of the tail shape and nonlinearity \eqref{subsecInfoScale}. The scale and other parameters of a non-exponential distribution are measured by the probability of $q$ random variables being in the state $x$ \eqref{subsec_IE}. This independent equals distribution, also known as the escort distribution, is the core innovation retained from $q-$statistics; however, $q$ is shown to be a secondary rather than a primary property of complex systems. The core inconsistency of the Tsallis $q-$statistics is shown to be a definition of the scale that is dependent on both linear and nonlinear sources of uncertainty \eqref{subsecEntropies}.  From these foundations, it is proven that the coupled entropy provides a unique balance in quantifying the uncertainty of a family of non-exponential distributions at the scale of a distribution and that this solution restricts the dependence on the nonlinear shape parameter to a generalization of the partition function \eqref{subsecCE}. Lemmas regarding the composability and extensivity of the coupled entropy are proven, which provide a basis for generalized entropy axioms \eqref{subsecAxioms}. 

Following the Results, the Discussion of Applications section \eqref{secDiscuss} describes the significance of the results by a) proving that the shape, coupling parameter $\kappa$ is to first-order a measure of complexity, b) proving that the coupled thermostatistics leads to a of the generalized temperature equal to the scale, $\sigma,$ and thus is independent of the nonlinearity, and c) reviewing how complex infodynamics impacts the design of artificial intelligence and communications systems for complex environments with heavy-tailed phenomena. Just as Goldilocks of fairy tale lore required a perfect temperature for her porridge, so to, scientists, mathematicians, and engineers will only be able to realize the analytical capabilities of a generalized infodynamics if there is perfection in its measurement of uncertainty.

Following the Discussion, the proof regarding the maximization of the coupled entropy by the coupled stretched exponential distributions is detailed \eqref{sec_maxCE}. Following the concluding remarks \eqref{sec_concl}, the Appendix includes a summary of the Generalized entropy functions  \eqref{app_Entropies}, a comparison with the Hanel-Thurner classification \eqref{app_Hanel}, and a description of the Mathematica Github repository \eqref{app_github}.

\section{Preliminaries} \label{sec_prelim}

\subsection{Functions for nonlinear analysis}\label{subsec_CPA}
The coupled exponential family \cite{nelson_definition_2015} provides a unification of many heavy-tailed (and compact-support) distributions while preserving the priority of properly defining the shape and scale of the distributions via the information scale requirement. Use of a generalized exponential function, which was popularized with Tsallis statistics \cite{borges_possible_2004}, while maintaining the informational scale, requires starting with the survival function; otherwise, the measure of dimensions $d$, which impacts the probability distribution function (PDF), is not properly separated. Next, it is important to distinguish between the shape of the distribution near the location $(\alpha: x\rightarrow\mu)$ and the asymptotic shape of the distribution $(\kappa: x\rightarrow\infty)$.   To do so, the exponent of the survival function (SF) is $-\sfrac{1}{\alpha\kappa}$, so that $\alpha$'s influence diminishes as $x\rightarrow\infty$, and the power law tail decay is determined by $\kappa$ alone. Likewise, for finite $\kappa$, as $x\rightarrow 0$ the survival function approximates a stretched exponential with $\alpha$ determining the curvature. 

The definition of the coupled exponential family uses a generalization of the exponential function, $\exp_\kappa^a(x) \equiv \exp_{\kappa/a}(a x)\equiv(1+\kappa x)^\frac{a}{\kappa}$. The inverse of the coupled exponential function is the coupled logarithm, $\frac{1}{a}\ln_\kappa x\equiv \ln_{a\kappa} x^\frac{1}{a}\equiv\frac{1}{a\kappa}(x^\kappa-1)$. For functions limited to the positive range, the symbol $(a)_+=\max(0,a)$ is used. These generalized functions have the following pseudo-algebraic properties \cite{borges_possible_2004,nelson_average_2017}:
\begin{align}
    \exp_\kappa(x+y)&=\exp_\kappa(x) \otimes_\kappa \exp_\kappa(y);\ 
    A \otimes_\kappa B \equiv (A^\kappa + B^\kappa-1)_+^\frac{1}{\kappa} \label{equ_cprod}\\
    \ln_\kappa (xy) &= \ln_\kappa (x) \oplus_\kappa \ln_\kappa (y); \ 
    A \oplus_\kappa B = A + B + \kappa AB. \label{equ_csum}
\end{align}
The coupled subtraction function follows from the coupled sum via the property
\begin{align}
    A\oplus_\kappa(\ominus_\kappa A)=0, \text{ therefore, }
    \ominus_\kappa A = \frac{-A}{1+\kappa A}, 
    \text{ for } A \neq \frac{-1}{\kappa} \\
    A\ominus_\kappa B=A\oplus_\kappa(\ominus_\kappa B)
    = \frac{A-B}{1+\kappa B}, 
    \text{ for } B \neq \frac{-1}{\kappa}.
\end{align}
The coupled division function follows from the coupled product, via the property:
\begin{align}
     A\otimes_\kappa(\oslash_\kappa A)=1, \text{ therefore, }
    \oslash_\kappa A = (2-A^\kappa)_+^\frac{1}{\kappa}, 
    \text{ for } A > 0 \\
    A\oslash_\kappa B=A\otimes_\kappa(\oslash_\kappa B)
    = (A^\kappa - B^\kappa+1)_+^\frac{1}{\kappa} , 
    \text{ for } A,B >0.
\end{align}
The methods extend to summations and products over $N$ variables are:
\begin{align}
\coupledsum{i=1}{N}{\kappa} x_i &= x_1 \oplus_\kappa x_i 
    ... \oplus_\kappa x_N\\
\coupledprod{i=1}{N}{\kappa} &= x_1 \otimes_\kappa x_i 
    ... \otimes_\kappa x_N \\
    &=\exp_\kappa\left(\sum_{i=1}^N\ln_\kappa x_i\right).
\end{align}
Finally, the coupled power is:
\begin{align}
    x^{\oplus_\kappa^N} &= (Nx^\kappa-(n-1))^\frac{1}{\kappa} \\
    &= \coupledprod{i=1}{N}{\kappa} x = x \otimes_\kappa x ... 
    \otimes_\kappa x \\
    &= (2x^\kappa -1) \otimes_\kappa x ... 
    \otimes_\kappa x = (3x^\kappa -2)^\frac{1}{\kappa} ... \oplus_\kappa x 
\end{align}
\subsection{The Coupled Exponential Family}\label{subsecCEF}
The coupled exponential family (CEF) of distributions is designed to generalize the exponential family $(\kappa=0)$ for heavy-tailed $(\kappa>0)$ and compact-support $(-\frac{1}{d}<\kappa<0)$ distributions, while assuring that a) the asymptotic shape of the distribution is defined solely by the coupling $\kappa$, and b) that other parameters in the family have precise interpretations for mathematical physics. The broadest definition of the family comes from a multivariate representation using information geometry. The definition is provided for continuous variables, though it extends to discrete distributions.
\begin{definition}[The Coupled Exponential Family]
    Given a $d-$dimensional random variable $\mathbf{X}_d \sim f_\kappa(\mathbf{x}_d;\boldsymbol{\theta},\alpha)$, within the coupled exponential family, its definition is:
    \begin{align}
        f_\kappa(\mathbf{x}_d;\boldsymbol{\theta},\alpha) &= \frac{h(\mathbf{x})}{Z_\kappa(\boldsymbol{\theta},\alpha,d)}
        \exp_{\alpha\kappa}^{-(1+d\kappa)}
        \left[\boldsymbol{\eta}(\boldsymbol{\theta}) \cdot \mathbf{T}(\mathbf{x})\right] \\
        &= \exp_{\alpha\kappa}^{-(1+d\kappa)}
        \left[\boldsymbol{\eta}(\boldsymbol{\theta}) \cdot \mathbf{T}(\mathbf{x}) \oplus_{\alpha\kappa} \ln_{\alpha\kappa}\left(h(\mathbf{x})\right)^{-\frac{1}{1+d\kappa}}
        \oplus_{\alpha\kappa} \ln_{\alpha\kappa}\left(Z_\kappa(\boldsymbol{\theta},\alpha,d)\right)^\frac{1}{1+d\kappa}
        \right], \label{equ_cefPDF}
    \end{align}
where $\alpha$ is the highest power of the variable function $\mathbf{T}(\mathbf{x})$, and $d$ is the dimension of the variable $\mathbf{x}$. $Z_\kappa$ is the normalization or partition function, and $h(\mathbf{x})$ is the base measure. 
\end{definition}
Notice that this definition of the coupled exponential family preserves the definition of the partition function as the integral of the other terms:
\begin{align}
    Z_\kappa(\boldsymbol{\theta},\alpha,d)= \int_{\mathbf{X}} 
    h(\mathbf{x}) \exp_{\alpha\kappa}^{-(1+d\kappa)}
        \left[\boldsymbol{\eta}(\boldsymbol{\theta}) \cdot \mathbf{T}(\mathbf{x})\right]
        \mathrm{d}\mathbf{x}
\end{align}
This property is assured by bringing $Z$ and possibly $h$ into the coupled exponential function via the coupled sum of the coupled logarithm. This contrasts with the definition for the $q-$exponential family defined by Ohara and Amari \cite{ohara_geometry_2007, amari_geometry_2011}:
\begin{align}
    f_q(\mathbf{x}_d;\boldsymbol{\theta}) &=h(\mathbf{x})
    \exp_q
        \left[\boldsymbol{\eta}(\boldsymbol{\theta}) \cdot \mathbf{T}(\mathbf{x})
        - \ln_qZ_q^{-1}(\boldsymbol{\theta})
        \right]\\
        \exp_q(x) &\equiv (1 + (1-q) x)^\frac{1}{1-q}; \ 
        \ln_q(x) \equiv \frac{1}{1-q}\left(x^{1-q} - 1\right), \label{equ_qexpln}
\end{align}
in which the term $Z_q$ is not the integral of the other terms and thus loses crucial properties regarding the partition function.
\begin{figure*}[ht]
    \centering
    \subfloat[]{
        \includegraphics[width=0.45\linewidth,page=1]{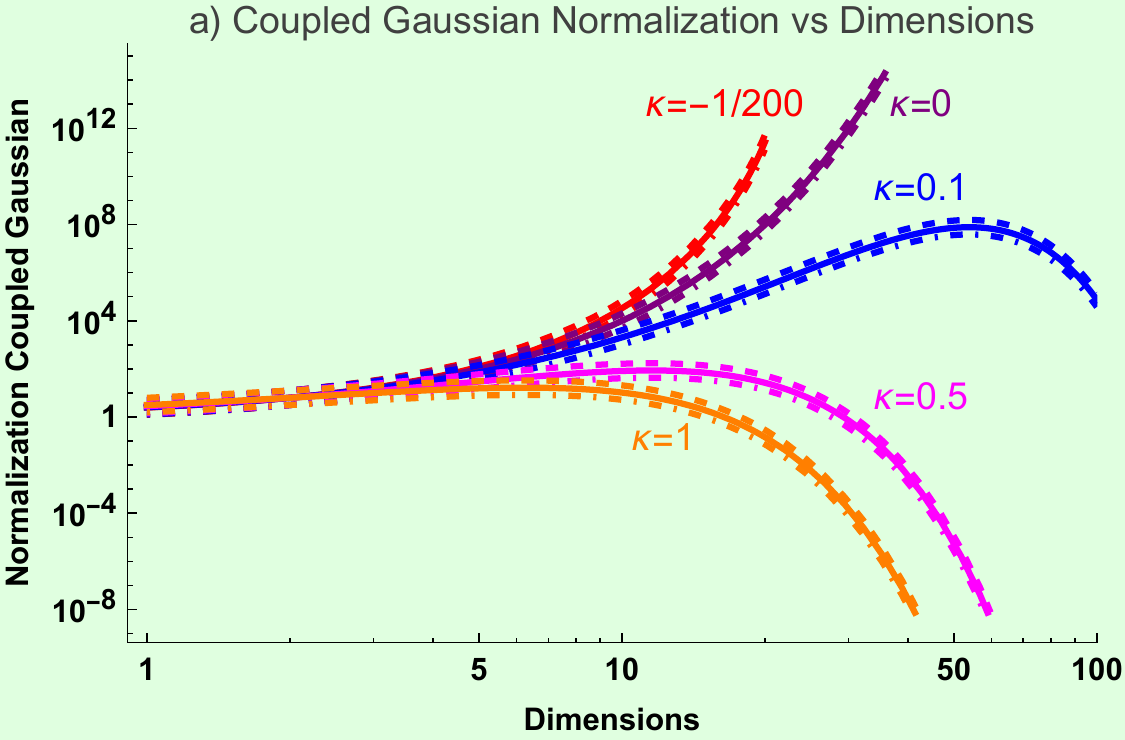}
        \label{fig_CG_Norm_a}
    }%
    \hfill
    \subfloat[]{
        \includegraphics[width=0.45\linewidth,page=1]{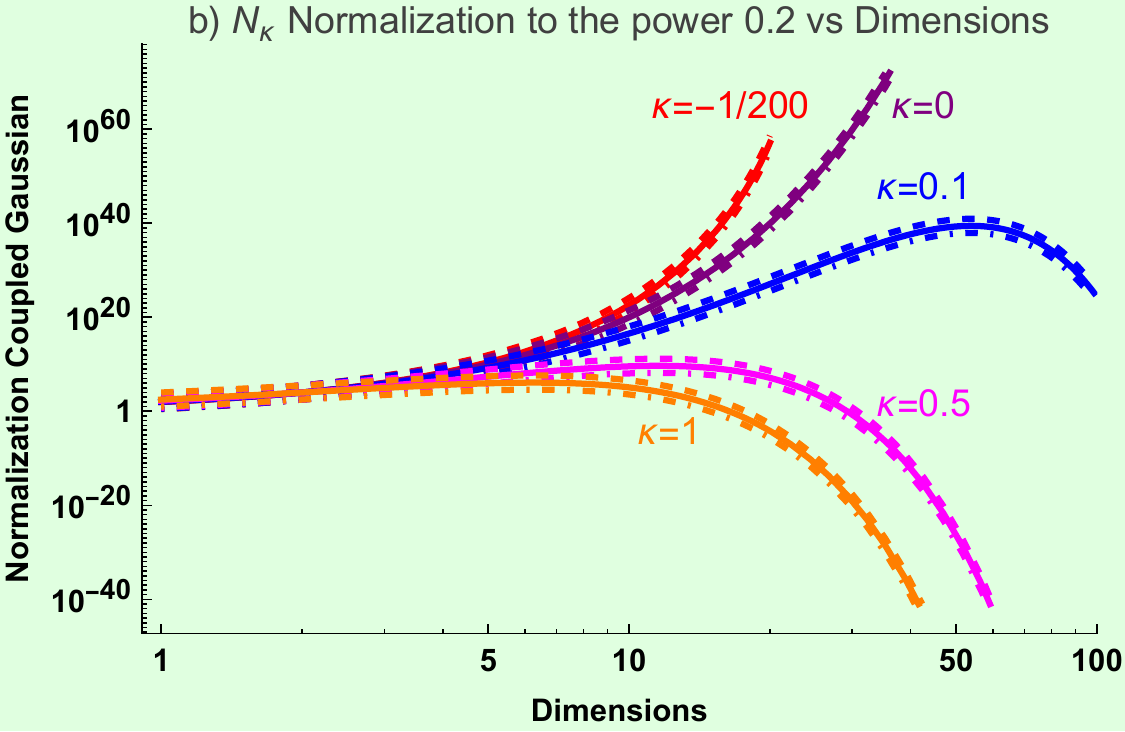}
        \label{fig_CG_Norm_b}
    }
    \caption{The Coupled Gaussian normalization as a function of dimensions, a) without modification b) raised to the power 1/5.}
    \label{fig_CG_Norm}
\end{figure*}

\begin{example}[Members of the CEF]
    Two broad groups of the CEF are when $h(\mathbf{x})=1$ (the coupled stretched exponential distribution) and when the survival function is only $\exp_{\alpha\kappa}^{-1}
        \left[\boldsymbol{\eta}(\boldsymbol{\theta}) \cdot \mathbf{T}(\mathbf{x})\right]$ (the coupled Weibull distribution). To simplify the expression, a radial variable with the stretching parameter $\sfrac{\alpha}{2}$ applied to elements of the vectors and matrix is defined as $r^\alpha\equiv ((\mathbf{x}-\boldsymbol{\mu})^{\circ \sfrac{\alpha}{2} \top} 
(\boldsymbol{\Sigma}^{\circ \sfrac{\alpha}{2}})^{-1}(\mathbf{x}-\boldsymbol{\mu})^{\circ \sfrac{\alpha}{2}}\geq 0$.

        \textbf{I. Coupled Stretched Exponential Distribution}
        PDF:
        \begin{align}
f_\kappa^{\text{Stretch}}(\mathbf{x}) &\equiv \frac{1}{Z_\kappa}\left(1+\kappa
r^\alpha\right)
^{-\frac{1+d\kappa}{\alpha\kappa}} \nonumber \\ \label{equ_csePDF}
&\equiv \frac{1}{Z_\kappa}\exp_{\alpha\kappa}^{-(1+d\kappa)}
\left(\frac{r^\alpha}{\alpha}\right) \nonumber \\
&\equiv \frac{1}{Z_\kappa}\exp_\kappa^{-\frac{1+d\kappa}{\alpha}}
\left(r^\alpha\right) \\
Z_\kappa&=\frac{2\pi^\frac{d}{2}}{\Gamma\left(\frac{d}{2}\right)}
|\boldsymbol{\Sigma}|^\frac{1}{2}
\left\{
\begin{matrix}
    \frac{1}{\alpha}\kappa^{-\frac{d}{\alpha}} B\left(\frac{d}{\alpha},\frac{1}{\alpha\kappa}\right)& \kappa > 0 \\
    \alpha^{\left(\frac{d}{\alpha}-1\right)}\Gamma\left(\frac{d}{\alpha}\right) & \kappa=0 \\
    \frac{1}{\alpha}(-\kappa)^{-\frac{d}{\alpha}} B\left(\frac{d}{\alpha},1-\frac{1+d\kappa}{\alpha\kappa}\right)& -\frac{1}{d}<\kappa<0 .
\end{matrix}
\right. \label{equ_cseNorm}\\
\nonumber \end{align}
        SF = 1 - CDF:
        \begin{align}
            S_\kappa^{\mathrm{Stretch}}(\mathbf{x})=1-F_\kappa(\mathbf{x}) 
            &\equiv \left\{
    \begin{matrix}
    \left(1-I_z\left(\frac{d}{\alpha},\frac{1}{\alpha\kappa}\right)\right)
     & \kappa > 0 \\
      \frac{\Gamma\left(\frac{d}{\alpha},\frac{r^\alpha}{\alpha}\right)}{\Gamma\left(\frac{d}{\alpha}\right)} & \kappa = 0 \\
     \left(1-I_z\left(\frac{d}{\alpha},1-\frac{1+d\kappa}{\alpha\kappa}\right)\right)
     & -1<\kappa< 0 
    \end{matrix}\right. \label{equ_cseSF} \\ 
    z&=\left\{
    \begin{matrix}
        \frac{\kappa r^\alpha}{1+\kappa r^\alpha} & \kappa > 0 \nonumber \\
        -\kappa r^\alpha & -\frac{1}{d}<\kappa<0
    \end{matrix}\right. \nonumber \\
    I_z&=\frac{B_z(a,b)}{B(a,b)};\text{  Regularized Incomplete Beta Function} \nonumber
        \end{align}
        Special Cases:
        $\alpha=1$: Coupled Exponential Distribution
        $\alpha=2$: Coupled Gaussian Distribution

        \textbf{II. Coupled Weibull Distribution}
        Assuming radial symmetry and starting with the survival function, which has just the coupled exponential structure.
        SF = 1 - CDF:
        \begin{align}
            S_\kappa^{\mathrm{Weibull}}(\mathbf{x})=1-F_\kappa(\mathbf{x}) \equiv
            \left\{\begin{matrix}
                \left(1+\kappa r^\alpha\right)_+^{-\frac{1}{\alpha\kappa}}
                & \kappa > -\frac{1}{d}; \ \kappa \neq 0\\ 
                \exp\left(-r^\alpha\right) & \kappa = 0
            \end{matrix}\right. \label{equ_cwSF}
        \end{align} 

         PDF: 
        \begin{align}
            f_\kappa^{\mathrm{Weibull}}(\mathbf{x}) 
            &\equiv \left\{\begin{matrix}
                \frac{1}{Z_\kappa}r^{\circ \frac{\alpha-1}{2}}
                \left(1+\kappa r^\alpha\right)_+^{-\frac{1}{\alpha\kappa}-d}
                & \kappa > -\frac{1}{d}; \ \kappa \neq 0\\
                \frac{1}{Z_\kappa}r^{\circ \frac{\alpha-1}{2}}
                \exp\left(-r^\alpha\right) & \kappa = 0
            \end{matrix}\right. \label{equ_cwPDF} \\ 
            Z_\kappa&=2\pi^\frac{d}{2}|\boldsymbol{\Sigma}|^\frac{1}{2}
\left\{
\begin{matrix}
    \left(\alpha\kappa^\frac{d}{\alpha}\Gamma(\frac{d}{2})\right)^{-1}
    B\left(\frac{d}{\alpha},\frac{1}{\alpha\kappa}+1 - \frac{d}{\alpha}\right)& \kappa > 0 \\
    \left(\alpha\Gamma(\frac{d}{2})\right)^{-1} \Gamma\left(\frac{d}{\alpha}\right) & \kappa=0 \\
    \left(\alpha(-\kappa)^\frac{d}{\alpha}\Gamma(\frac{d}{2})\right)^{-1} B\left(\frac{d}{\alpha},-\frac{1}{\alpha\kappa}-\frac{d}{\alpha}+1\right)& -\frac{1}{d}<\kappa<0 .
\end{matrix}
\right. \nonumber \\
\nonumber
        \end{align}
        Special Cases:
        $\alpha=1$: Coupled Exponential Distribution
        $\alpha=2$: Coupled Rayleigh Distribution

        The one-dimensional coupled exponential distribution is unique in retaining just the coupled exponential structure for both SF and the PDF.
\end{example}

\subsection{Multiplicative Noise Fluctuations}\label{subsec_noise}
The fluctuations of multiplicative noise provide further evidence that the structure of the coupled stretched exponentials provides a unique model of the mathematical physics. Beck \cite{beck_superstatistics_2003} showed that the gamma distribution model of fluctuation of the scale of an exponential distribution is distributed as a $q-$exponential distribution and that several other fluctuation distributions have a $q-$exponential as a second-order approximation. This model can be extended to the full stretched exponentials, with $\beta=\sigma^{-\alpha}$, in which case,  $\sigma^{-\alpha}$ is the mean and $\kappa=\frac{\left<\sigma^{-2\alpha}\right>-\left<\sigma^{-\alpha}\right>^2}{\left<\sigma^{-\alpha}\right>^2}$ is the relative variance of the fluctuations \cite{nelson_average_2017}. 

Likewise, a multiplicative noise process with a parabolic potential relating the deterministic and multiplicative functions has a non-equilibrium stationary state (NESS) of a coupled Gaussian, $X_{NESS} \sim N_\kappa(X_0,\sigma),$ \cite{al-najafi_independent_2024, anteneodo_multiplicative_2003}. The Stratonovich stochastic equation in terms of the NESS scale $\sigma$ and coupling $\kappa$ is: 
\begin{equation}\label{equ_multproc}
    dX_t=-X_tdt+ \sqrt{2}\sigma \circ dW_t^{(a)}+ \sqrt{2\kappa} X_t^2 \circ dW_t^{(m)}.
\end{equation}
Figure \ref{fig:MultProc} shows several samples of the process for $\sigma=5$ and a set of couplings $\kappa= (0.1, 1, 10)$. While the coupled Gaussian scale is independent of the coupling, which creates fluctuations in the process, a $q$-Gaussian model has a scale of $\sqrt{\beta^{-1}}=\frac{\sigma}{\sqrt{1+\kappa}}$. The inverse dependence on the coupling undermines the ability of $\sqrt{\beta^{-1}}$ to model a generalized scale and physical properties, such as a generalized temperature. Antendeodo \cite{anteneodo_multiplicative_2003} also defined the multiplicative noise process with a stationary distribution that generalizes the stretched exponential distribution. The mapping to the CESD is not as direct, so additional research in defining those processes is recommended.
\begin{figure}[ht]
    \centering
    \includegraphics[width=1\linewidth,page=1]{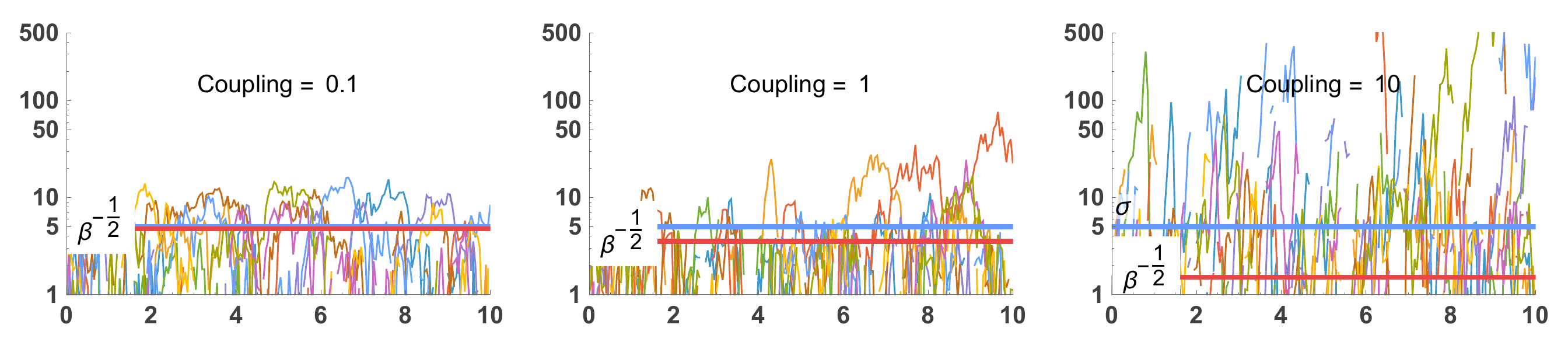}
\caption{\textbf{Multiplicative Process Samples} 10 samples from the multiplicative noise process defined by equation \eqref{equ_multproc}, showing the positive values on a logarithmic scale. The NESS distribution is a coupled Gaussian with $\sigma=5$ and a set of couplings $\kappa=(0.1,1,10)$. As the coupling increases, the fluctuations of the process increase, while the scale $\sigma=5$ is independent of the fluctuations. In contrast, the $q$-Gaussian scale $\sqrt{\beta^{-1}}$  is dependent on the multiplicative noise, which undermines its ability to be a measure of generalized temperature. }
    \label{fig:MultProc}
\end{figure}

\section{Main Results}
\label{secResults}

\subsection{Uniqueness of the Informational Scale}\label{subsecInfoScale}

Specifying the scale of a shape-scale distribution would appear to be an elementary task, but, in fact, the lack of a clear criteria has resulted in different statistical physics communities using different criteria. For instance, Tsallis statistics \cite{shalizi_cosma_r_tsallis_2021} uses a scale derived from the hypothesized Tsallis entropy, while the space plasma \cite{pierrard_kappa_2010} community specifies the average velocity as a scale. I'll prove that the only scale that separates the linear and nonlinear sources of uncertainty is the value at which the derivative of the log-log of the distribution is negative one. Given the connection to the surprisal of the distribution, I'll refer to this as the \textit{informational scale}. 

Recall that a shape-scale type distribution can be scaled if given a standard (non-scaled) random variable, $X \sim f(x),$ the scaled random variable is $\sigma X \sim \frac{1}{\sigma}f(\frac{x}{\sigma})$. While this is a criterion for the distribution, it does not specify the definition of the scale, since $\sigma'=a\sigma+b$ also satisfies the relationship. The information scale requirement will be specified in reference to the coupled exponential family, which is defined in the Methods section \ref{subsecCEF}.

\begin{definition}[Informational Scale]
    The informational scale, $\sigma$, of a distribution with shape parameters $(\alpha,\kappa)$ and location, $\mu$, is the value of $x$ such that:
    \begin{align}\label{equ_info_scale}
        \frac{\mathrm{d}\ln f(x-\mu;\sigma,\kappa,\alpha)}{\mathrm{d}\ln x}\vert_{x=\mu+\sigma}=x \frac{f'(x-\mu;\sigma,\kappa,\alpha)}{f(x-\mu;\sigma,\kappa,\alpha)}\vert_{x=\mu+\sigma}=-1
    \end{align}
    Equivalently, using the surprisal rather than the derivative of the log-log, the criterion is:
    \begin{align}
        \frac{\mathrm{d}\ln f(x-\mu;\sigma,\kappa,\alpha)}{\mathrm{d} x}\vert_{x=\mu+\sigma}=-\frac{1}{\sigma}
    \end{align}
\end{definition}

I will show that the parameterization for the generalized Pareto (P) \cite{arnold_pareto_2015} and the Student's t (S) \cite{pearson_students_1942} distributions satisfy the informational scale, while the Tsallis $q$-exponential (Texp) and $q$-Gaussian (TGauss), as well as the space plasma kappa $(\kappa_P)$ distribution, do not. The Tsallis and space plasma definitions seek to generalize thermodynamic principles, where the $q$-statistics inverse-scale, $\beta$, is proportional to a generalized inverse-temperature, and the plasma scale of $\frac{2k_B T}{m}$, reflects an energy due to velocity ($k_B:$ Boltzmann constant;  $T:$ temperature;  $m:$ mass). Without loss of generality, the location is set to zero.

\begin{example}[Informational \lowercase{scale of common heavy-tailed distributions}]
    \begin{align*}
        \text{Distribution Definitions}\\
        f_P(x;\sigma,\kappa)&=\frac{1}{\sigma}\left(1+\kappa\frac{x}{\sigma}\right)^{-\frac{1}{\kappa}-1}\\
        f_{Texp}(x;\beta,q)&=(2-q)\beta\left(1+(q-1)\beta x\right)^{\frac{1}{1-q}}\\
        f_S(x;\sigma,\nu)&=\frac{1}{Z_S}\left(1+\frac{x^2}{\nu\sigma^2}\right)^{-\frac{1}{\alpha}(\nu+1)}\\
        f_{TGauss}(x;\beta,q)&=\frac{1}{Z_{TG}}\left(1+(q-1)\beta x^2\right)^{\frac{1}{1-q}}\\
        f_{\kappa_P}(x;\frac{2k_B T}{m},\kappa_P)&=\frac{1}{Z_{\kappa_P}} \left[ 1 + \frac{m x^2}{2k_BT (\kappa - 3/2)} \right]^{-\kappa - 1}\\
    \end{align*}
    \begin{align*}
        \text{Solutions for }&\text{Informational Scale } (\sigma_I^{dist})\\
        \sigma_I^P&=\sigma\\
        \sigma_I^{Texp} &=\frac{1}{(2-q)\beta}\\
        \sigma_I^S &=\sigma\\
        \sigma_I^{TGauss}&=\frac{1}{\sqrt{(3-q)\beta}}\\
        \sigma_I^{\kappa_P} &=\frac{2k_BT}{m}\frac{\kappa_P-\sfrac{3}{2}}{\sqrt{2\kappa_P+1}}
    \end{align*}
\end{example}

Throughout the paper, several properties dependent on the information scale will be proven. First, the informational content gained by clearly separating the scale and shape of non-exponential distributions is shown by the Boltzmann-Gibbs-Shannon (BGS) entropy, $H(f(x))=-\int_{x\in X}f(x)\ln f(x) \mathrm{d}x$, of the example distributions.
\begin{align}
    H(f_P(x))&=1+\ln\sigma +\kappa\\ \label{equ_ent_GPD}
    H(f_{Texp}(x))&=1+\ln\frac{1}{(2-q)\beta} +\frac{q-1}{2-q}\\
    H(f_S(x))&=\ln Z_S +\frac{\nu + 1}{2}\left(\psi(\frac{\nu+1}{2})-\psi(\frac{\nu}{2})\right);\ \psi \text{ is the digamma function}\\
    H(f_{TGauss}(x))&=1+\ln Z_{TG} +\frac{1}{q-1}\left(\psi(\frac{1}{q-1})-\psi(\frac{q-1}{2(3-q)})\right)\\
    H(f_{\kappa_P}(x))&=1+\ln Z_{TG} +(\kappa_P+1)\left(\psi(\kappa_P+1)-\psi(\kappa_P+\sfrac{1}{2})\right)\\
\end{align}

The structure of the entropy for the shape-scale distributions is a constant plus the logarithm of the partition function (normalization) plus a function of the shape. The entropy of the GPD equation \eqref{equ_ent_GPD} has a direct simplicity. In Shannon's classic paper on the theory of communication \cite{shannon_mathematical_1948}, he states after reviewing the entropy axioms, "The real justification for these definitions will reside in their implications." In this vein, the uniqueness proof for the Coupled Entropy will begin with a pragmatic derivation and a proof that its solution for a broad range of shape-scale distributions, $f(\mathbf{x};\sigma,\kappa,\alpha,d),$ has a simple structure, whose dependence on the asymptotic shape, $\kappa$, is contained within a generalized logarithm, $\ln_g x$, of the partition function, Z:
\begin{equation}
H_\kappa(f(\mathbf{x};\sigma,\kappa,\alpha,d),\alpha,d)
=g_1(\alpha,d)+\ln_{g(\kappa,\alpha,d)}Z(\sigma,\kappa,\alpha,d);
\end{equation}
With such a solution, the measurement of uncertainty is focused on the partition function, much like the exponential family. Furthermore, I'll prove that this solution is associated with the density at the location plus the scale.

For clarity, the proof of a unique scale is completed for the one-dimensional case.  The two most important distributions are the coupled exponential and the coupled Gaussian.

\begin{definition}[Coupled Exponential Distribution]\label{def_CED}
For location $\mu$, scale $\sigma$, shape $\alpha=1$,  and $d=1,$ the survival function for the one-sided coupled exponential distribution is:
\begin{align} 
\text{SF: } S_\kappa^\text{exp}(x;\mu,\sigma,1)
\equiv \left(1 + \kappa\left(\frac{x-\mu}{\sigma}\right)\right)_+^{-\frac{1}{\kappa}}
\equiv\exp_{\kappa}^{-1}
\left(\frac{x-\mu}{\sigma}\right); \\ 
x>\mu; \ \kappa > -1 \nonumber
\end{align}

The probability density function, $-\frac{\mathrm{d}S}{\mathrm{d}x},$ for the one-sided coupled exponential distribution is:
\begin{align} 
\text{PDF: }f_\kappa^{\text{exp}}(x;\mu,\sigma,1) &\equiv \frac{1}{\sigma}\left(1+\kappa\frac{x-\mu}{\sigma}\right)_+^{-\frac{1+\kappa}{\kappa}}\label{equ_cexpdist} 
\equiv \frac{1}{\sigma}\exp_\kappa^{-(1+\kappa)}\left(\frac{x-\mu}{\sigma}\right), \\  \label{equ_cepdf}
\text{for } x\geq \mu, \kappa>-1.\nonumber
\end{align}
\end{definition}

\begin{definition}[Coupled Gaussian Distribution]\label{def_CGD}
For location $\mu$, scale $\sigma$, shape $\alpha=2$, and $d=1,$ the survival function for the coupled Gaussian distribution is:
 \begin{align}
            S_\kappa(x;\mu,\sigma,2)=1-F_\kappa(x) 
            &\equiv \left\{
    \begin{matrix}
    \left(1-I_z\left(\frac{1}{2},\frac{1}{2\kappa}\right)\right)
     & \kappa > 0 \\
      \frac{1}{2}\left(1-\text{erf}\left(\frac{x-\mu}{\sqrt{2}\sigma}\right)\right) & \kappa = 0 \\
     \left(1-I_z\left(\frac{1}{2},\frac{-1+\kappa}{2\kappa}\right)\right)
     & -1<\kappa< 0 
    \end{matrix}\right.\\
    z&=\left\{
    \begin{matrix}
        \frac{\kappa \left(\frac{x-\mu}{\sigma}\right)^{2}}{1+\kappa \left(\frac{x-\mu}{\sigma}\right)^{2}} & \kappa > 0 \nonumber \\
        -\kappa \left(\frac{x-\mu}{\sigma}\right)^{2} & -1<\kappa<0
    \end{matrix}\right. \nonumber \\
    I_z&=\frac{B_z(a,b)}{B(a,b)};\text{  Regularized Incomplete Beta Function} \nonumber
        \end{align}

The probability density function, $-\frac{\mathrm{d}S}{\mathrm{d}x},$  for the coupled Gaussian distribution is:
\begin{align} 
f_\kappa^{\text{exp}}(x;\mu,\sigma,2) &\equiv \frac{1}{Z}\left(1+\kappa\left(\frac{x-\mu}{\sigma}\right)^2\right)_+^{-\frac{1+\kappa}{2\kappa}}\label{equ_CGdist} \\
&\equiv \frac{1}{Z}\exp_{2\kappa}^{-(1+\kappa)}\left(\frac{1}{2}\left(\frac{x-\mu}{\sigma}\right)^2\right), 
\ \text{for } \kappa>-1.\\ \label{equ_cgpdf}
\end{align}
\end{definition}

\begin{lemma}[Independent Properties of the CESD Parameters]\label{lem_CESD}
   Given a random variable X distributed as the CSED pdf equation \eqref{equ_cexpdist}, then 
   \begin{enumerate}[label=\bfseries\Roman*.]
       \item the coupling parameter, $\kappa$, is the only shape parameter of the distribution when $|x-\mu|\gg\sigma$ and the asymptotic distribution is the Type I Pareto or a pure power law;
       \item the location shape parameter, $\alpha$, is the only shape parameter of the distribution when $|x-\mu|\ll\sigma$ and the distribution converges to a stretched exponential distribution;
       \item the scale parameter, $\sigma$, is exclusively the informational scale as defined by equation \eqref{equ_info_scale}.
   \end{enumerate}
\end{lemma}
\begin{proof}
    $ $\\
   \begin{enumerate}[label=\bfseries\Roman*.]
   \item 
   \begin{align*}
        \lim_{|x|\gg\mu+\sigma}\frac{1}{Z}\left(1+\kappa\left|\frac{x-\mu}        {\sigma}\right|^\alpha\right)^{-\frac{1+\kappa}{\alpha\kappa}}
        &=\frac{1}{Z}\left(\kappa\left|\frac{x-\mu}{\sigma}\right|^\alpha\right)^{-\frac{1+\kappa}{\alpha\kappa}}\\
        &=\frac{\alpha\kappa^{-\frac{1}{\alpha\kappa}}}
        {\sigma B\left(\frac{1}{\alpha},\frac{1}{\alpha\kappa}\right)}
       \left|\frac{x-\mu}{\sigma}\right|^{-\frac{1+\kappa}{\kappa}}
    \end{align*}
     Therefore, for $|x-\mu|\gg\sigma$  is a pure power law with exponent, $-\frac{1}{\kappa}-1$. For a $d-$dimensional distribution, the $-1$ factor is $-d$. Therefore, the asymptotic shape of the CESD is only determined by $\kappa$.
    \item For $|x-\mu|\ll\sigma$ and a fixed $\kappa$, $\kappa\left|\frac{x-\mu}{\sigma}\right|\ll 1$. Therefore  $\kappa\ll \left|\frac{\sigma}{x-\mu}\right|$, which is the same requirement for the convergence of the generalized exponential to the exponential function.  Therefore,
    \begin{align*}
        f_\kappa^{\text{Stretch}}(x) \approx \frac{\alpha}{\sigma \Gamma\left(\frac{1}{\alpha}\right)}
        \exp\left(-\frac{1}{\alpha}\left|\frac{x-\mu}{\sigma}\right|^\alpha\right) 
        \text{ when  } |x-\mu|\ll\sigma.
    \end{align*}
    \item Without loss of generality, assuming $\mu=0$,
    \begin{align*}
        -\frac{\mathrm{d}\ln\left[\frac{1}{Z}\left(1+\kappa\left|\frac{x}{\sigma}\right|
        ^\alpha\right)^{-\frac{1+\kappa}{\alpha\kappa}}\right]}{\mathrm{d}\ln(x)} 
        &= -x \frac{\partial_x\left[\frac{1}{Z}\left(1+\kappa\left|\frac{x}{\sigma}\right|
        ^\alpha\right)^{-\frac{1+\kappa}{\alpha\kappa}}\right]}{\left[\frac{1}{Z}\left(1+\kappa\left|\frac{x}{\sigma}\right|
        ^\alpha\right)^{-\frac{1+\kappa}{\alpha\kappa}}\right]}\\
        &=\frac{(1+\kappa)\left|\frac{x}{\sigma}\right|^\alpha}{1+\kappa\left|\frac{x}{\sigma}\right|^\alpha},
    \end{align*}
    which for $x=\pm\sigma$ is equal to 1. Therefore, $\sigma$ is the informational scale.
   \end{enumerate}

\end{proof}

The contrast between the independence of the scale of the coupled exponential and the dependence of the $q-$exponential distribution on the shape is shown in Figure \ref{fig_scale}.
\begin{figure*}[ht] 
    \centering
    \subfloat{
        \includegraphics[width=0.475\linewidth,page=1]{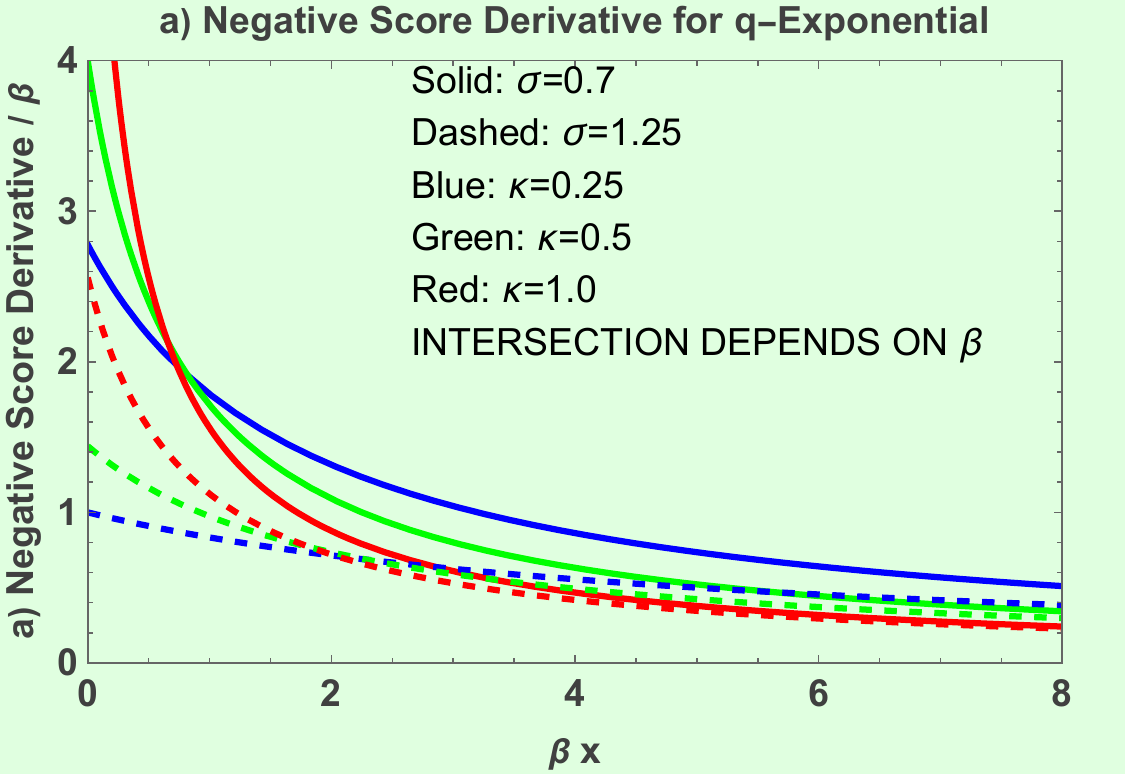}
        \label{fig_scale_a}
        }
    \hfill
    \subfloat[]{
        \includegraphics[width=0.475\linewidth,page=1]{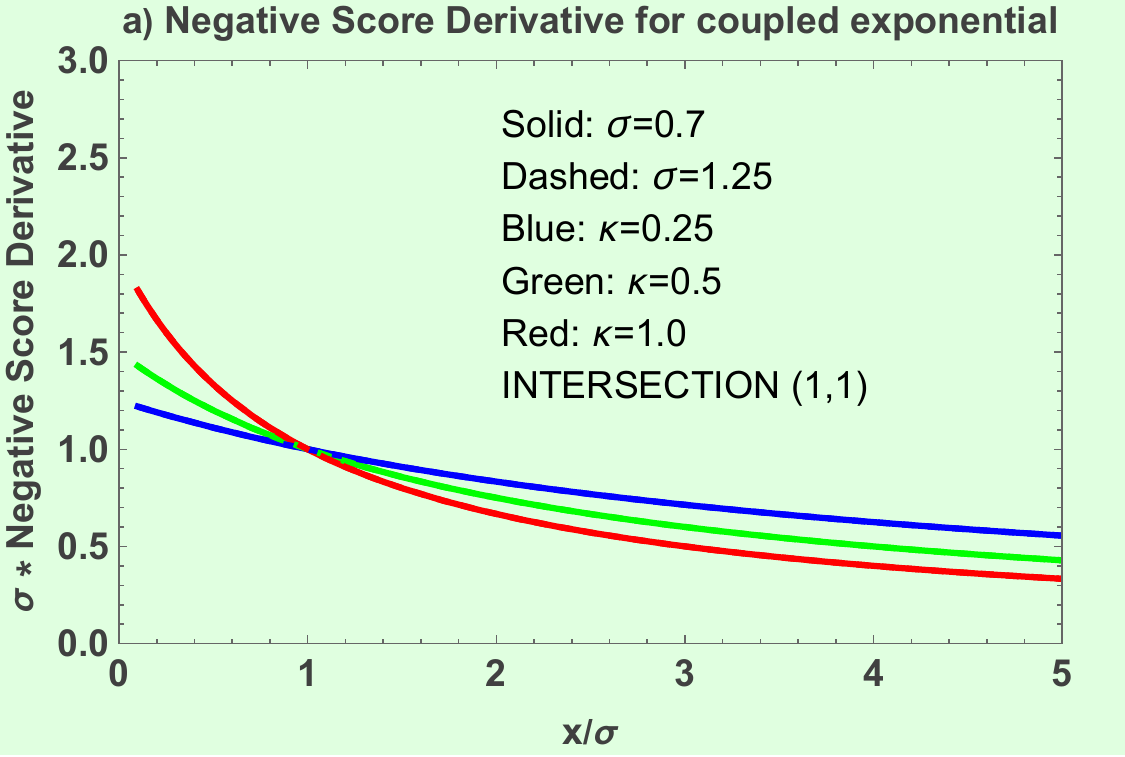}
        \label{fig_scale_b}
        }
    \caption{The negative derivative of the score function (logarithm of the distribution) shows the uniqueness of the information scale, $\sigma.$ a) The inverse-scale $(\beta)$  of the $q-$exponential does not have a common intersection. b) The information scale $(\sigma),$  when normalized, has a common intersection independent of the shape $\kappa.$ The $-\sigma \ \mathrm{d}\ln(f(x))/\mathrm{d}x$ function is not dependent on $\sigma.$ } 
    \label{fig_scale}
\end{figure*}

\subsection{Required Solution for a Generalized Entropy}\label{subsec_required}

Prior to any definition for a generalized entropy, the coupled exponential family provides guidance regarding the required solution. This comes from understanding the structure of the entropy for the stretched exponential members of the exponential family with $h(x)=1$ and considering $d=1$. In this case, the BGS entropy is:
\begin{align}
    &H(f((x-\mu)^\alpha/\alpha;\sigma)) \nonumber \\ 
    &=-\int_{x\in\mathcal{X}} f((x-\mu)^\alpha/\alpha;\sigma)\ln[
    \exp(-\frac{(x-\mu)^\alpha}{\alpha \sigma^\alpha}-\ln Z(\sigma,\alpha)
    ]dF(x) \\
    &=\frac{1}{\alpha}+\ln Z(\sigma,\alpha). \nonumber
\end{align}
An infrequently recognized but important property of the entropy is that this is the negative logarithm of the density at the location plus the scale (at the scale, for short), 
\begin{align}
    H[f((x-\mu)^\alpha/\alpha;\sigma)]
    &= -\ln[f(\mu+\sigma; \mu,\sigma,\alpha)] \nonumber \\
    &=\alpha^{-1}+\ln Z(\sigma,\alpha)
\end{align}
 A generalized entropy is required to preserve the structure of this solution. That is, the generalized entropy of a coupled stretched exponential distribution must be equal to the density at the scale transformed to the entropy domain with the inverse of the generalized exponential. Given Lemma \ref{lem_CESD} establishing the requirement for the coupled exponential family so that the scale is independent of the shape, the required solution for what will be called the coupled entropy is established. For simplicity, the solution is shown for one dimension here; however, the $d-$dimensional solution will be derived in Section \ref{subsecCE}.
\begin{definition}[Required Generalized Entropy Solution] \label{def_required} 
\begin{enumerate}
    \item Given a 1-D coupled stretched exponential distribution, \eqref{equ_csePDF}, 
    \begin{equation}
        f_\kappa(\mu+\sigma;\mu,\sigma,\alpha)
        = \exp_{\alpha\kappa}^{-(1+\kappa)}
        \left(\frac{1}{\alpha} \oplus_{\alpha\kappa} 
        \ln_{\alpha\kappa}
        Z_\kappa(\sigma,\alpha)^\frac{1}{1+\kappa}
        \right)
    \end{equation}
    then, a generalized entropy is required to have the following solution:
    \begin{align}
    H_\kappa^{Required}(f_\kappa(x;\mu,\sigma,\alpha),\alpha) 
    &= \ln_{\alpha\kappa}\left(f_\kappa(\mu+\sigma;\sigma,\alpha)
    ^{-\frac{1}{1+\kappa}}\right) \\
    &=\frac{1}{\alpha} \oplus_{\alpha\kappa} 
    \ln_{\alpha\kappa}Z_\kappa(\sigma,\alpha)^\frac{1}{1+\kappa} \\
    &=\frac{1}{\alpha} + (1 + \kappa)\ln_{\alpha\kappa}
    Z_\kappa(\sigma,\alpha)^\frac{1}{1+\kappa} \\
    &= \frac{1}{\alpha} + \ln_\frac{\alpha\kappa}{1+\kappa}
    Z_\kappa(\sigma,\alpha).
    \end{align}
    \item For the extension to $d-$dimensions, the notation $\mathbf{x}^{\circ \frac{\alpha}{2}}$ indicates that each element of the vector or matrix  is raised to the power $\sfrac{\alpha}{2}.$ The $d-$D coupled stretched exponential and its required generalized entropy solution are:
    \begin{align}
        f_\kappa((\boldsymbol{\mu+\sigma})
        ^{\circ \frac{\alpha}{2}};
        \boldsymbol{\mu}^{\circ \frac{\alpha}{2}},
        \boldsymbol{\Sigma}^{\circ \frac{\alpha}{2}})
        &= \exp_{\alpha\kappa}^{-(1+d\kappa)}
        \left(\frac{d}{\alpha} \oplus_{\alpha\kappa} 
        \ln_{\alpha\kappa}
        Z_\kappa(\sigma,\alpha)^\frac{1}{1+d\kappa}
        \right) \\
        H_\kappa^{Required}(f_\kappa((\boldsymbol{\mu+\sigma})
        ^{\circ \frac{\alpha}{2}};
        \boldsymbol{\mu}^{\circ \frac{\alpha}{2}},
        \boldsymbol{\Sigma}^{\circ \frac{\alpha}{2}}),\alpha,d) 
    &= \ln_{\alpha\kappa}\left(
    f_\kappa((\boldsymbol{\mu+\sigma})
        ^{\circ \frac{\alpha}{2}};
        \boldsymbol{\mu}^{\circ \frac{\alpha}{2}},
        \boldsymbol{\Sigma}^{\circ \frac{\alpha}{2}})
    ^{-\frac{1}{1+d\kappa}}\right) \\
    &=\frac{d}{\alpha} \oplus_{\alpha\kappa} 
    \ln_{\alpha\kappa}Z_\kappa(\sigma,\alpha)^\frac{1}{1+d\kappa} \\
    &=\frac{d}{\alpha} + (1 + d\kappa)\ln_{\alpha\kappa}
    Z_\kappa(\sigma,\alpha)^\frac{1}{1+d\kappa} \\
    &= \frac{d}{\alpha} + \ln_\frac{\alpha\kappa}{1+d\kappa}
    Z_\kappa(\sigma,\alpha).
    \end{align}
\end{enumerate}
\end{definition}
\begin{remark}
    This is the only solution that associates the generalized entropy with the density at the scale, thereby minimizing the dependence on the nonlinear statistical coupling. Furthermore, in the Discussion section on thermodynamics \ref{subsec_thermo}, it will be shown that this is the precision structure required to assure that the generalization of the thermodynamic fulfills the principle that the entropy $(S)$ is equal to the energy $(U)$ per temperature $(T)$ plus the logarithm of the partition function, $S=\frac{U}{kT}+\ln Z,$ where $k$ is the Boltzmann constant. 
\end{remark}

The distinction between the required solution and the most common entropies, BGS, Rényi, Tsallis, and Normalized Tsallis, for the coupled exponential distribution is shown in Figure \ref{fig_entDensity}. Each point on a density curve represents an entropy translated to the density domain via the function $\exp_\kappa^{-(1+\kappa)}(H(f(x)).$ The generalized entropies and their translation to a density value are reviewed in Appendix \ref{app_Entropies}.

For the BGS entropy, which has a linear dependence on the shape and a logarithmic dependence on the scale, as the shape increases, its equivalent density value is further away from the scale.  The disconnect between the BGS entropy and the scale of a non-exponential distribution is the essence of why a generalized entropy is required for the uncertainty quantification of complex systems. Rényi took the first important step by replacing the geometric mean with the generalized mean, and Tsallis further improved this by utilizing a generalized logarithm. I'll prove in Section \ref{subsecCE} that the requirement to measure the uncertainty at the scale is fulfilled by the coupled entropy. 

A key element of the required entropy solution is that the expression $\frac{(x-\mu)^\alpha}{\alpha\sigma^\alpha}$ must be multiplied by a modified distribution such that its integral, equivalent to a generalization of the $\alpha-$moment, is equal to $\sfrac{1}{\alpha}$ even for non-exponential distributions. This is fulfilled by the Independent Equals distribution, known in the nonextensive statistical mechanics literature as the escort distribution. The details of this component of the coupled entropy function and its constraints will be reviewed next.  
\begin{figure}
    \centering
    \includegraphics[width=0.75\linewidth,page=1]{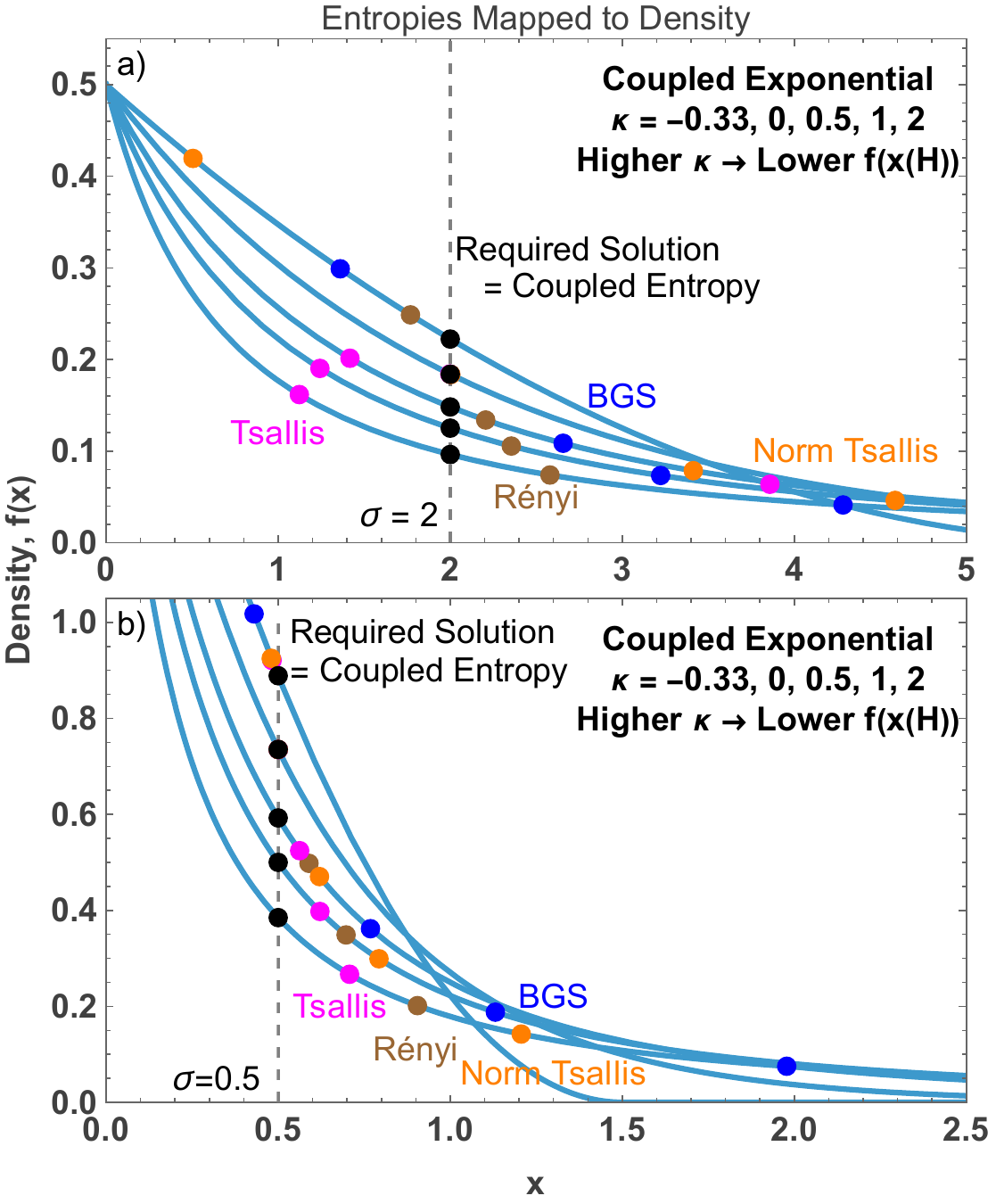}
    \caption{Entropies of the coupled exponential distribution are mapped to points on the density. Coupling values of -0.33, 0, 0.5, 1, and 2 are displayed. The required solution satisfied by the coupled entropy aligns along the scale, a) $\sigma=2,$ and b) $\sigma=0.5.$ BGS with higher entropy corresponds lower density values and a higher value of $x.$ Rényi lower the entropy measure for heavy-tailed distributions via use of the generalized mean. Tsallis further improved the measure via use of a generalized logarithm. Although the normalized Tsallis is structurally closer to the correct solution its high measure of entropy is unsuitable. All of the entropies converge for $\kappa=0.$}
    \label{fig_entDensity}
\end{figure}

\subsection{Independent Equals Moments} \label{subsec_IE}
In the heavy-tailed domain, $\kappa>0,$ the moments, $\mu_m=\int_{x\in\mathcal{X}}x^mp(x)dF(x),$ of the coupled exponential distributions are either undefined or divergent for $\kappa\geq\sfrac{1}{m}.$ Nevertheless, the escort probability \cite{beck_thermodynamics_1993,ferri_role_2005}, has been shown to define the distribution of independent random variables that share the same state (\textit{independent-equals}) \cite{nelson_independent_2022,al-najafi_independent_2024}. When the moment matches the location shape, $\alpha$, the measure is the power of the informational scale, $\sigma^\alpha$. Throughout the text the Lebesgue–Stieltjes integral is used to provide a concise expression of both discrete and continuous distributions.
\begin{lemma}[Independent Equal Moments]\label{def_IEM}
\begin{enumerate}
    \item Given a distribution with asymptotic tail shape, $\kappa$, the independent equals distribution with power $q=1+\frac{m\kappa}{1+\kappa},$
\begin{equation} \label{equ_IE}
    f/^{1+\frac{m\kappa}{1+\kappa}} \equiv 
    \frac{f(x)^{1+\frac{m\kappa}{1+\kappa}}}{\int_{x\in X}f(x)^{1+\frac{m\kappa}{1+\kappa}}dF(x)} ,
\end{equation}
has a modified shape of $\kappa'=\frac{\kappa}{1+m\kappa}$ and thus finite moments for $m'< \frac{1+m\kappa}{\kappa}$, which always includes $m$ for finite $\kappa$.
    \item Given $X\sim f_\kappa(x;\mu,\sigma,\alpha)$ distributed as a CSED equation \eqref{equ_csePDF} the independent equals moment in which $m=\alpha$ is equal to the informational scale raised to the power $\alpha$:
\begin{align}
    \mu_m^{(1+\frac{m\kappa}{1+\kappa})} &\equiv E_{(1+\frac{m\kappa}{1+\kappa)}}\left[X^m\right]
    \equiv \int_X x^m f/^{1+\frac{m\kappa}{1+\kappa}}dF(x) \\
    \mu_\alpha^{(1+\frac{\alpha\kappa}{1+\kappa})} &= \sigma^\alpha
\end{align}
\item Given $X\sim f_\kappa(x;\mu,\sigma,\alpha)$ distributed as a coupled Weibull distribution, equation \eqref{equ_cwPDF}, the $\alpha$ moment and its independent equals moment are:
    \begin{align}
        \mu_\alpha^{(1)}&=\mathrm{E}[X^\alpha]=\sigma^\alpha,
        \text{ for } -1<\kappa<\frac{1}{\alpha}\\
        \mu_\alpha^{(1+\frac{\alpha\kappa}{1+\kappa})}&=\mathrm{E}_{1+\frac{\alpha\kappa}{1+\kappa}}[X^\alpha]=\frac{\sigma^\alpha}{1+\alpha\kappa},
        \text{ for } \kappa>-\frac{1}{\alpha}
    \end{align}
\end{enumerate}
\end{lemma}
\begin{proof}
    \begin{enumerate}
        \item The asymptotic distribution is a Type I Pareto distribution (pure power-law), $f(x) \sim C x^{-(\frac{1}{\kappa}+1)}$ as $x\rightarrow\infty$, and therefore, the independent equals distribution is $F^{(1+\frac{m\kappa}{1+\kappa})}\sim C'x^{-\frac{1+(1+m)\kappa}{\kappa}}$. Therefore,
         \begin{align}
             \frac{1 + \kappa '}{\kappa '}
                &= \frac{1 + m\kappa + \kappa}{\kappa} \\ 
                \frac{1}{\kappa'}&= \frac{1}{\kappa} + m \\
                \kappa'&=\frac{\kappa}{1+m\kappa}
         \end{align}
          and thus, the finite moments are $m_{finite}<\frac{1}{\kappa'}=m+\frac{1}{\kappa}$.
        \item Given that the asymptotic shape of the independent equals distribution of the coupled stretched exponential is modified to $\kappa'=\frac{\kappa}{1+\alpha \kappa}$ and the ratio $\frac{\kappa}{\sigma^\alpha}$ is unchanged, then $\sigma'^\alpha=\frac{\sigma^\alpha}{1+\alpha\kappa}$. Therefore,
    \begin{align}
        \mu_\alpha^\alpha&=\int_X x^\alpha f_{\kappa'}(\mu,\sigma',\alpha)\mathrm{d}x=\sigma'^\alpha\frac{1}{1-\alpha\kappa'}\\
        &=\frac{\sigma^\alpha}{1+\alpha\kappa}\frac{1+\alpha\kappa}{1}=\sigma^\alpha.
    \end{align}
    \item Likewise, the independent equals distribution of the coupled Weibull has a modified shape and scale of $\kappa'=\frac{\kappa}{1+\alpha \kappa}$ and $\sigma'^\alpha=\frac{\sigma^\alpha}{1+\alpha\kappa}$, respectively. Given that the coupled Weibull distribution approximates a Weibull distribution for $x\rightarrow0$, and a Type-I Pareto distribution for $x\rightarrow\infty$, then  
    \begin{align}
        \mu_\alpha^{(1+\frac{\alpha\kappa}{1+\kappa})}&=\int_0^\infty x^\alpha f_{\kappa'}(0,\sigma',\alpha)\mathrm{d}x\\
        &=\int x^\alpha\frac{\alpha}{\sigma'}\left(\frac{x}{\sigma'}\right)^{\alpha-1} \left(\alpha\kappa'\left(\frac{x}{\sigma'}\right)^{-(\frac{1}{\kappa'}+\alpha)} \right)\mathrm{d}x\vert_{x\rightarrow\infty} \\
        & \ \ \ \ -\int x^\alpha \frac{\alpha}{\sigma'}\left(\frac{x}{\sigma'}\right)^{\alpha-1} \exp\left(-\left(\frac{x}{\sigma'}\right)^\alpha\right) \mathrm{d}x\vert_{x=\mu}\\
        &=\int C x^{-\frac{1}{\kappa}-1}\mathrm{d}x\vert_{x\rightarrow\infty}
        -\left(-\sigma'^\alpha\Gamma\left(2,-\left(\frac{x}{\sigma'}\right)^\alpha\right)\vert_{x=0} \right) \\
        &=0+\sigma'^\alpha=\frac{\sigma^\alpha}{1+\alpha\kappa}
        \text{ for } \kappa \geq -\frac{1}{\alpha}.
    \end{align}
    Notice that the independent equals moment is proportional to the modified scale. Therefore, the unmodified moment is proportional to the scale but over a restricted heavy-tailed domain:
    \begin{align}
        \mu_\alpha^{(1)}&=\int_0^\infty x^\alpha f_\kappa(0,\sigma,\alpha)\mathrm{d}x\\
        &=\int x^\alpha\frac{\alpha}{\sigma}\left(\frac{x}{\sigma}\right)^{\alpha-1} \left(\alpha\kappa\left(\frac{x}{\sigma}\right)^{-(\frac{1}{\kappa}+\alpha)} \right)\mathrm{d}x\vert_{x\rightarrow\infty} \\
       & \ \ \ \  -\int x^\alpha \frac{\alpha}{\sigma}\left(\frac{x}{\sigma}\right)^{\alpha-1} \exp\left(-\left(\frac{x}{\sigma}\right)^\alpha\right) \mathrm{d}x\vert_{x=\mu}\\
        &=\int C x^{\alpha-\frac{1}{\kappa}-1}\mathrm{d}x\vert_{x\rightarrow\infty}
        -\left(-\sigma^\alpha\Gamma\left(2,-\left(\frac{x}{\sigma}\right)^\alpha\right)\vert_{x=0} \right) \\
        &=0+\sigma^\alpha
        \text{ for } -1<\kappa < \frac{1}{\alpha}.
    \end{align}
    
    \end{enumerate}
\end{proof}

\subsection{Entropies of the Coupled Exponential Family}\label{subsecEntropies}
Over the past sixty years, a considerable breadth of research in statistical physics and information theory has focused on generalizations of the BGS entropy suitable for analysis of complex systems in which non-exponential distributions maximize a generalized entropy. In the 1960s, Alfred Rényi \cite{renyi_measures_1961} proposed that the generalized mean replace the geometric mean, which composes the probabilities for the BGS entropy. The coupled exponential family is the maximizing distribution for Rényi entropy, given a match between the power of the generalized mean and the exponent of the distribution. While the Rényi entropy transposes the generalized mean of the probabilities to an entropy with the logarithmic function, Constantino Tsallis \cite{tsallis_possible_1988, tsallis_nonadditive_2009-1} took this construction further by proposing a generalized logarithm for the transposition between the probability domain and the entropy domain. In doing so, Tsallis broke the additivity of entropy for composed subsystems. While the Rényi entropy is additive, the composition of the Tsallis entropy includes a nonlinear term. Because the Tsallis entropy uses the same generalized mean as the Rényi entropy, it is also maximized by the coupled exponential family. Nevertheless, I'll show that both the Rényi and the Tsallis entropy have a defect in their definition of the logarithmic transposition that undermines their entropic solution.

In parallel with the search for entropy functions relevant for the analysis of nonlinear systems has been a theoretical maturation of the minimum set of axioms  needed to specify the uniqueness of these functions. If, as I suggest, there is a missing requirement in the form of a generalized entropy, there must also be a modified or additional requirement for the axioms. Just as Euclid's first four geometric postulates are foundational primitives, the first three Shannon-Khinchin axioms \cite{shannon_mathematical_1948,khinchin_mathematical_1957} are necessary for all proposed entropies with a few exceptions. The first three axioms are:
\begin{enumerate}
    \item \textbf{Continuity}: Entropy, $H$, only depends on the probabilities of a distribution, $H(\mathbf{p})=H(p_1,p_2,...,p_W),$ where $W$ is the number of states in the system and $\textbf{p}=(p_1,p_2,...p_W)$ is a normalized probability distribution.
    \item \textbf{Maximality}: Given only the constraint of a normalized distribution, the uniform distribution maximizes the entropy, $H(\frac{1}{W}, \frac{1}{W},...,\frac{1}{W})\geq H(p_1,p_2,...,p_W)$.
    \item \textbf{Expandability}: A state with $p_i=0$ does not change the entropy, $H(p_1,p_2,...,p_W,0)$$=H(p_1,p_2,...,p_W)$.
\end{enumerate}

Assuming just these three axioms, Hanel and Thurner \cite{hanel_comprehensive_2011,Hanel2011a} used two different scaling properties to define a broad family of (c,d)-entropies. The parameter $c$ is equivalent to the generalization parameter for the Rényi and Tsallis entropies, $q=1+\frac{\alpha\kappa}{1+d\kappa}.$ For $q=1$ the (1,d)-entropy is maximized by the stretched exponential distributions; however, the definition utilized differs from the common expression $\exp\left(-\frac{x^\alpha}{\alpha\sigma^\alpha}\right).$ Further discussion of the coupled entropies relationship with Hanel, Thurner classification is in Appendix \ref{app_Hanel}.

Tempesta \cite{tempesta_beyond_2016, tempesta_universality_2020} defined a universal class of entropies that are composable with a power series of nonlinear terms. Thus, just as Euclid's fifth axiom can be generalized to define a variety of Riemannian geometries, the fourth Shannon-Khinchin axiom generalizes to define non-additive but composable entropies. In fact, the information geometry of the resulting metrics induces non-Euclidean geometries \cite{amari_geometry_2011}.

    4. \textbf{Composability}: Tempesta's universal entropy, $H^U(\textbf{p})\equiv \sum_i p_i G\left(\ln p_i^{-1}\right)$, where the power series composition function is $G(t) \equiv \sum_{k=0}^\infty a_k \frac{t^{k+1}}{k+1}$, uniquely satisfies the following composition axiom. Given two statistically independent systems with probabilities, $p_i^I$ and $p_j^{II}$, the joint entropy of the two systems is: 
\begin{align}
    H^U(\textbf{p}^I,\textbf{p}^{II}) = G\left(G^{-1}(H^U(\textbf{p}^I)) + G^{-1}(H^U(\textbf{p}^I))\right)
\end{align}
where $G^{-1}(s)$ is the compositional inverse of G(t).

The composition function for the Tsallis entropy converges to a function with one nonlinear term. Furthermore, the Tsallis' entropy is composable for non-independent systems:
    \begin{alignat}{2}
        H^T\left(\left\{p_{ij}^{I,II}\right\}\right)
        &=H^T\left(\left\{p_{i}^{I}\right\}\right)
        +\sum_i p_i^I H^T\left(\left\{p^{II}(j|i)\right\}\right) \\
        &+ (1-q) H^T\left(\left\{p_{i}^{I}\right\}\right) 
        \sum_i p_i^I H^T\left(\left\{p^{II}(j|i)\right\}\right) \\
        &=H(p_1,p_i,...,p_{W_i})
        +\sum_i p_i^I H^T\left(\frac{p_{i1}}{p_i},...\frac{p_{iW_j}}{p_i}\right) \\
        &+ (1-q) H^T(p_1,p_i,...,p_{W_i})
        \sum_i p_i^I H^T\left(\frac{p_{i1}}{p_i},...\frac{p_{iW_j}}{p_i}\right).
    \end{alignat}
where $p_i^I$ is the probability of the first system, and $p_j^{II}$ is the probability of the second system. $p^{II}(j|i)$ is the joint probability of the second system given a particular state of the first system. The Shannon-Khinchin additivity axiom is recovered for $q=1$.

Both the Rényi and Tsallis entropies are trace form in which the generalized mean of the probabilities is computed and then the natural and generalized logarithms are applied, respectively. I will parameterize these entropies using the relationship $q=\alpha^{R\Acute{e}nyi}=1+\frac{\alpha\kappa}{1+d\kappa}$  with $d=1,$ in order to analyze the functions with relation to the coupled exponential family. 
\begin{align}
    \textbf{Rényi Entropy:  } H_\kappa^R(f(\mathbf{x}); \alpha)&= -\ln\left(\int_{\mathbf{x}\in\mathcal{X}} f(\mathbf{x})f(\mathbf{x})^\frac{\alpha\kappa}{1+\kappa}
    dF(\mathbf{x})\right)^\frac{1+\kappa}{\alpha\kappa} \\
    \textbf{Tsallis Entropy:  } H_\kappa^T(f(\mathbf{x}); \alpha)&= 
    - \int_{\mathbf{x}\in\mathcal{X}} f(\mathbf{x})
     \ln_\frac{\alpha\kappa}{1+d\kappa} f(\mathbf{x})
     dF(\mathbf{x})\\
    &=-\ln_\frac{\alpha\kappa}{1+\kappa}\left(
    \int_{\mathbf{x}\in\mathcal{X}} 
    f(\mathbf{x}) f(\mathbf{x})^\frac{\alpha\kappa}{1+\kappa}dF(\mathbf{x})\right)^\frac{1+\kappa}{\alpha\kappa}\\
    \textbf{Normalized Tsallis Entropy:  } \nonumber \\
    H_\kappa^{NT}(f(\mathbf{x}); \alpha)
    &=\int_{\mathbf{x\in\mathcal{X}}}
     f/^\frac{1+(d+\alpha)\kappa}{1+d\kappa}(\mathbf{x})
     \ln_\frac{\alpha\kappa}{1+d\kappa} 
     \left(f(\mathbf{x})^
     {-\frac{1+(d+\alpha)\kappa}{\alpha\kappa}}\right)
     dF(\mathbf{x}) \\
    &=\ln_\frac{\alpha\kappa}{1+\kappa}
    \left(
      \int_{\mathbf{x\in\mathcal{X}}}
      f/^\frac{1+(d+\alpha)\kappa}{1+d\kappa}(\mathbf{x})
      f(\mathbf{x})^{-\frac{1+(d+\alpha)\kappa}{1+d\kappa}} dF(\mathbf{x})\right)^
      {-\frac{1+d\kappa}{1+(d+\alpha)\kappa}}, \\ 
   \text{where } f^{\frac{1+\alpha\kappa}{1+\kappa}}(\mathbf{x}) &\equiv \frac{f(\mathbf{x})^{(\frac{1+\alpha\kappa}{1+\kappa})}}{\int_{\mathbf{x}\in\mathcal{X}}
   f(\mathbf{x})^\frac{1+\alpha\kappa}{1+\kappa}}
\end{align}
  \begin{figure}[ht]
      \centering
      \includegraphics[width=0.75\linewidth,page=1]{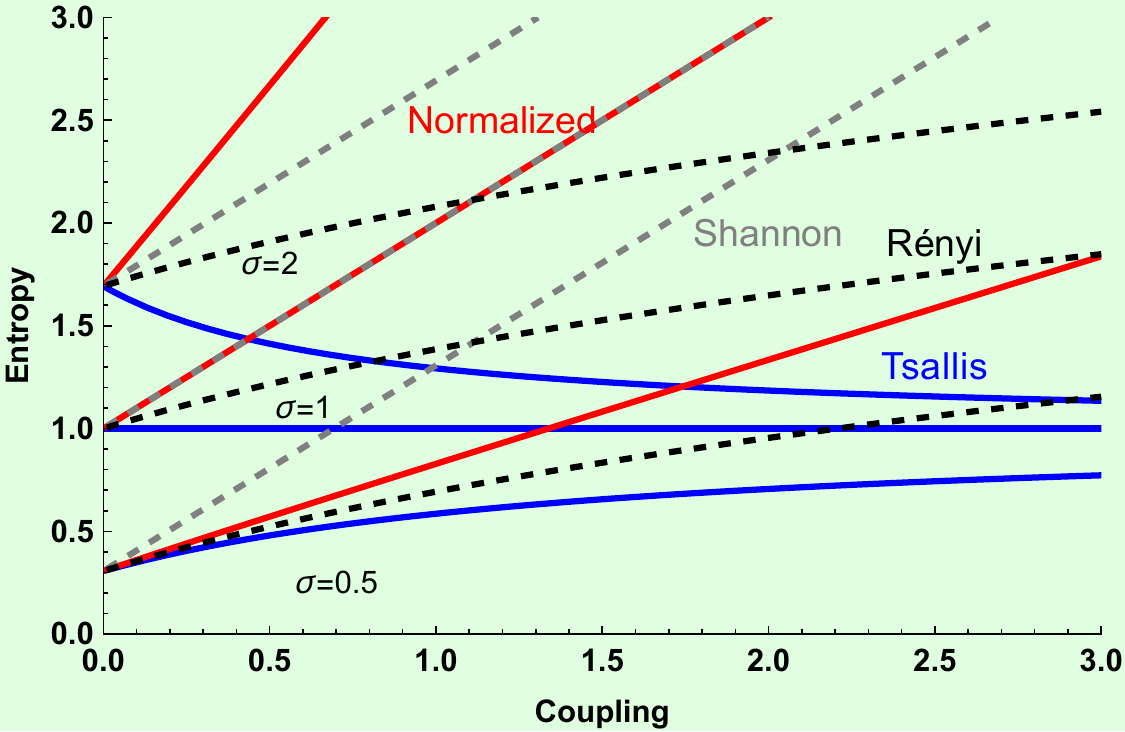} 
      \caption{\textbf{Entropies of the Coupled Exponential Distribution} The entropy versus the coupling $(\kappa)$ is shown for the Shannon (gray, dashed), Rényi (black, dashed), Tsallis (blue), and Normalized Tsallis (red) entropies. Scale $(\sigma)$ values of 0.5, 1, and 2 are shown. Shannon is linear and Renyi is logarithmic with the coupling. Both are logarithmic with the scale. Neither the Tsallis or Normalized Tsallis entropies provide a consistent metric of the scale of the distribution. Tsallis has an inverse relationship with the scale, converging to 1, and the Normalized Tsallis multiplies the scale term with the coupling.} \label{fig_entropies}
  \end{figure}
\begin{table}
    \centering
     \caption{Entropies of the Coupled Exponential Distribution}
    \label{tab_entropies}
    \begin{tabular}{ccc}\toprule
        \textbf{Entropy} & \textbf{of Coupled Exp Dist.} & \textbf{Description}\\\midrule
        \textbf{\textit{BGS}}& $1+\ln(\sigma)+\kappa$ & Logarithmic in Scale; Linear in Shape \\
        \textbf{\textit{Rényi}} & $\ln\sigma+(1+\frac{1}{\kappa})\ln(1+\kappa)$ & Logarithmic in Scale \& Shape\\
        \textbf{\textit{Tsallis}} & $1-\frac{1}{1+\kappa}\ln_\frac{\kappa}{1+\kappa} {\sigma^{-1}}$ & Inverse in Scale \& Shape\\
        \textbf{\textit{Normalized Tsallis}} & $1+(1+\kappa)\ln_\frac{\kappa}{1+\kappa}{\sigma}+\kappa$ & Multiplicative in Scale \& Shape\\
        \textbf{\textit{Required: Coupled}}&  $1+\ln_\frac{\kappa}{1+\kappa} \sigma$& Generalized Log of Scale\\ \bottomrule
    \end{tabular}
   
\end{table}

In Figure (\ref{fig_entropies}) and Table  (\ref{tab_entropies}) the entropy solution for the coupled exponential distribution (CED) equation \eqref{equ_cexpdist} is plotted and listed, respectively. The BGS entropy of the CED is logarithmic in the scale and linear in the shape. The linear dependence on the shape overwhelms the entropy for distributions with a large shape parameter; thus, there are analytic advantages for generalized entropies that reduce or remove the shape dependence. Indeed, the Rényi entropy changes the shape dependence to logarithmic, while maintaining the logarithmic scale dependence. The Tsallis entropy does eliminate the linear dependence on the shape; unfortunately, the generalized logarithm is of the inverse of the scale rather than the scale. Worse, the scale term includes an inverse dependence on the shape, such that as $\kappa\rightarrow\infty$ the Tsallis entropy of the CED goes to 1. The Normalized Tsallis entropy corrects the direct dependence on the generalized logarithm of the scale, but creates a more significant problem. The NTE has both a linear dependence on the shape and the shape multiplies the scale term, causing very rapid growth for $\sigma>1$. 

Next, I'll show that the Coupled Entropy is unique in fulfilling the requirement that the generalized entropy be a measure of the uncertainty at the scale.

\subsection{Uniqueness of the Coupled Entropy}\label{subsecCE}
The inadequacy of the Tsallis entropy's solution for the coupled exponential distribution provides a cautionary warning about axiomatic uniqueness proofs. While the Tsallis entropy was proven to be maximized by the $q-$exponential and $q-$Gaussian distributions, these were new definitions for non-exponential distributions with inadequate justification as to why they should replace the Pareto and Gosset (Student) definitions. In turn, the Suyari \cite{suyari_generalization_2004} and Furuichi \cite{furuichi_uniqueness_2005} "uniqueness" theorems did not so much prove that the Tsallis entropy was unique but rather designed axioms specifying the Tsallis entropy under the assumption that it had already demonstrated sufficient merit to justify an axiomatic definition.

A full proof that the coupled entropy is maximized by the CSEDs is completed in Section  \ref{sec_maxCE}. Here, I will review a heuristic derivation of the coupled entropy \cite{nelson_average_2017}, though a derivation using either variational methods \cite{wang_probability_2008, hanel_generalized_2012} , or the extensivity and group-theoretic requirements specified by Tempesta \cite{tempesta_universality_2020} would be a valuable contribution. First, the structure of the CSEDs equation \eqref{equ_csePDF} , which has a different multiplicative term $(\alpha\kappa)$ than its exponential term $\left(-\frac{1+d\kappa}{\alpha\kappa}\right)$, dictates that the coupled logarithm for the coupled entropy inverts these terms, $\ln_{\alpha\kappa} x^{-\frac{1}{1+d\kappa}}=$$\frac{1}{\alpha\kappa}\left(x^{-\frac{\alpha\kappa}{1+d\kappa}}-1\right).$ Second, putting aside the partition function (normalization) for the moment, the inversion of the coupled exponential leaves an average over the multivariate argument $\left(\frac{1}{\alpha}\left(\mathbf{x}^\top \mathbf{\Sigma}^{-1} \mathbf{x}\right)^\frac{\alpha}{2}\right)$.  If this average is taken over the distribution, then the resulting $\alpha-$moment will be dependent on the shape and only finite for $\kappa<\frac{1}{\alpha}.$ Thus, the independent equals moment of equation  \eqref{def_IEM} is required. Each of these moments over the $d^2$ terms cancels its $\Sigma_{ij}^{-1}$ term, leaving $\frac{d}{\alpha}.$ There is a third issue regarding the dependence of the solution on the power $u^\frac{\alpha}{d}$, which will be addressed as an optional, non-trace power of $\frac{1}{\gamma}.$

\begin{definition}[The Coupled Entropy]
    Given a $d-$dimensional random variable $\mathbf{X}$ over the sample space, $\mathcal{X}$, a coupling parameter that constitutes the asymptotic shape of the coupled exponential family, $\kappa > -\frac{1}{d},$ a location shape, $\alpha>0$, and a non-trace power factor, $\gamma>0$, the coupled entropy is defined as: 
    {\footnotesize
    \begin{align}
        & H_\kappa(\mathbf{X};\alpha,d,\gamma) \nonumber \\
        &= \left[\left\{
        \begin{matrix}
        \int_{\mathbf{x}\in\mathcal{X}}
        f/^\frac{1+(d+\alpha)\kappa}{1+d\kappa}(\mathbf{x})
        \ln_{\alpha\kappa}f^{-\frac{1}{1+d\kappa}}(\mathbf{x})
        dF(\mathbf{x})
        & \kappa \neq 0 \\
         -\int_{\mathbf{x}\in\mathcal{X}}
         f(\mathbf{x})\ln f(\mathbf{x})dF(\mathbf{x})
         & \kappa=0
        \end{matrix}\right.
        \right]^\frac{1}{\gamma}. \nonumber \\
        &= 
        \left[\left\{
        \begin{matrix}
        \ln_{\alpha\kappa}\left(
        \int_{\mathbf{x}\in\mathcal{X}}
        f/^{\frac{1+(d+\alpha)\kappa}{1+d\kappa}}(\mathbf{x})
        f/^\frac{\alpha\kappa}{1+d\kappa}(\mathbf{x})
        dF(\mathbf{x}) \right)
        ^{-\frac{1}{\alpha\kappa}}
        & \kappa \neq 0 \\
        -\int_{\mathbf{x}\in\mathcal{X}}
         f(\mathbf{x})\ln f(\mathbf{x})dF(\mathbf{x}) 
         & \kappa=0
        \end{matrix} \label{equ_CE}
        \right.\right]^\frac{1}{\gamma} \nonumber \\ 
    \end{align} 
    \nonumber } 
The first set of expressions shows the generalized mean (geometric mean for $\kappa=0$) within the coupled logarithm, while the second set of expressions shows the arithmetic average outside the coupled logarithm. The equivalence of the expressions is derived in Appendix \ref{app_Entropies}. For discrete distributions, the generalized logarithm of the generalized mean converges to the logarithm of the geometric mean, though that form cannot be expressed with integrals.
\end{definition}
\newpage

\begin{remark}[Structure \lowercase{o}f Coupled Entropy]
\leavevmode\newline
    \begin{itemize}
        \item The exponent of $\frac{1}{\gamma}$ is defined outside the coupled logarithm and summation, which follows the definition of \cite{tempesta_multivariate_2020, tempesta_universality_2020}, since preservation of the structure of the coupled logarithm is critical to achieving a unique solution for the coupled stretched exponentials.
        \item While the coupled entropy shares some similarities with the Tsallis entropy, the critical distinctions are:
        \begin{itemize}
            \item The divisive normalization of the coupled logarithm, $\frac{1}{\alpha\kappa},$ differs from the exponent of the probability, $-\frac{\alpha\kappa}{1+d\kappa}.$ The source of this distinction is the difference in exponents for the sf and pdf of the coupled exponentials.
            \item The average is over the independent equals distribution, assuring measurable moments for all coupling values.
        \end{itemize}
    \end{itemize}
\end{remark}

\begin{theorem}[Uniqueness: Coupled Entropy \lowercase{o}f Coupled Stretched Exponentials]
    \label{lem_unique}
    Given a $d-$dimensional random variable, $\mathbf{X} \sim f_\kappa(\mathbf{x}; \boldsymbol{\mu},\boldsymbol{\Sigma},\alpha,d)$, where$ f_\kappa(\mathbf{x}; \mu,\sigma,\alpha,d)$, or $f_\kappa(\mathbf{x})$ for short, is the CSED equation \eqref{equ_csePDF}, then the coupled entropy with matching values of the $\kappa,\alpha,\text{and }d,$ and $\gamma=1$ is 
\begin{equation}
    H_\kappa (\mathbf{X};\alpha,d,1)=\frac{d}{\alpha} + \ln_\frac{\alpha\kappa}{1+d\kappa}Z_\kappa(\sigma,\alpha,d).
\end{equation} 
This is the density of the distribution at the radius $r=\left(((\mathbf{x}-\boldsymbol{\mu})^{\circ \sfrac{\alpha}{2} \top} 
(\boldsymbol{\Sigma}^{\circ \sfrac{\alpha}{2}})^{-1}(\mathbf{x}-\boldsymbol{\mu})^{\circ \sfrac{\alpha}{2}}\right)^\frac{1}{\alpha}=d$ which for example occurs when $x_i=\mu_i+\sigma_{ii}$ for all $i.$ Thus, the measure of  the dependence on the shape of the distribution is constrained to the partition function and the generalized logarithm. Furthermore, the coupled logarithm parameter is constrained to $0<\frac{\alpha\kappa}{1+d\kappa}<\frac{\alpha}{d}$ for the heavy-tailed domain, $0<\kappa<\infty.$ 
\end{theorem}
\begin{proof}
    Without loss of generality, $\mu$ is set to zero. For $\kappa\neq 0$,
    \begin{align}
        H_\kappa(f_\kappa(\mathbf{x}))&=\bigintssss_\Omega F^{(\frac{1+(d+\alpha)\kappa}{1+d\kappa})}(\mathbf{x})
        \ln_{\alpha\kappa}f^{-\frac{1}{1+d\kappa}}(\mathbf{x})\mathrm{d}\mathbf{x}\\
        \ln_{\alpha\kappa}f^{-\frac{1}{1+d\kappa}}(\mathbf{x})\mathrm{d}\mathbf{x}
        &= \ln_{\alpha\kappa}\left[
        \frac{1}{Z_\kappa}\exp_{\alpha\kappa}^{-(1+d\kappa)}\left(
        \frac{1}{\alpha^\frac{2}{\alpha}}
        \mathbf{x}^\top \Sigma^{-1} \mathbf{x}\right)^\frac{\alpha}{2}
        \right]^{-\frac{1}{1+d\kappa}}
    \end{align}
Utilizing the coupled sum equation \eqref{equ_csum} to separate the partition term, $Z_\kappa^\frac{1}{1+d\kappa},$ from the coupled exponential term, there is the following simplification:
\begin{align}
    \ln_{\alpha\kappa}f^{-\frac{1}{1+d\kappa}}(\mathbf{x})\mathrm{d}\mathbf{x}
    &= \left(\ln_{\alpha\kappa}Z_\kappa^\frac{1}{1+d\kappa}\right)
    \oplus_{\alpha\kappa}\left(\frac{1}{\alpha}\left(
    \mathbf{x}^\top \Sigma^{-1} \mathbf{x}\right)^\frac{\alpha}{2}\right)
\end{align}
\begin{figure*}[ht] 
    \centering
    \subfloat[]{
        \includegraphics[width=0.45\linewidth,page=1]{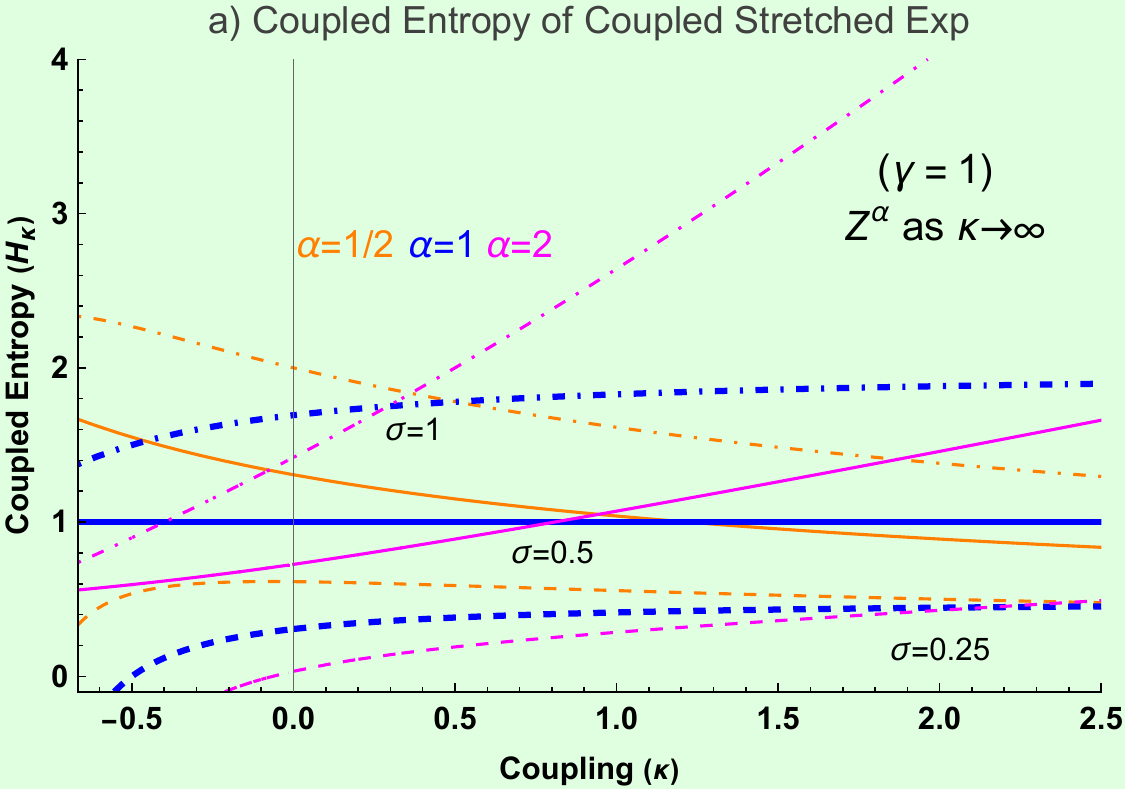}
        \label{fig_CE_CSE_a}
    }%
    \hfill
    \subfloat[]{
        \includegraphics[width=0.45\linewidth,page=1]{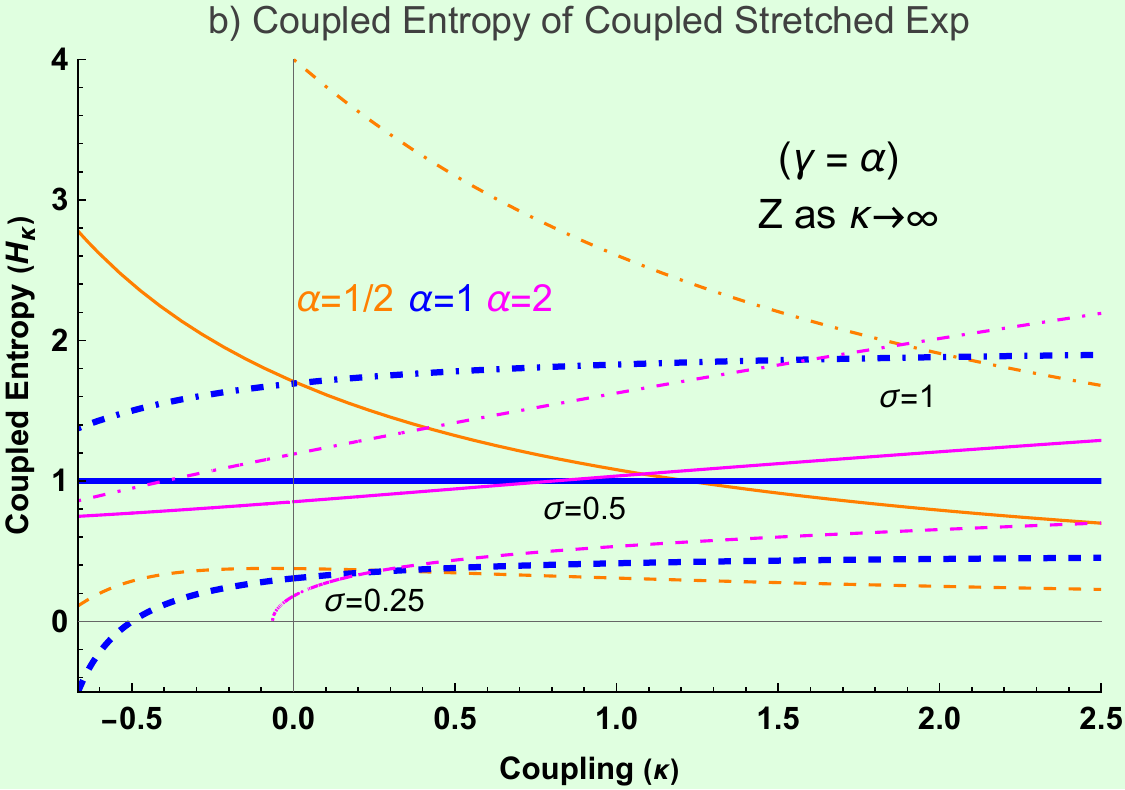}
        \label{fig_CE_CSE_b}
    }
    \caption{The Coupled Entropy of the double-sided Coupled Stretched Exponential Distribution as a function of the coupling $(\kappa)$. a) With $\gamma=1$ the coupled entropy converges to $Z^\alpha$ as $\kappa \rightarrow \infty.$ b) With $\gamma=\alpha$ the coupled entropy converges to $Z$ as $\kappa \rightarrow \infty.$ Each graph shows three values of the variable power, $\alpha$, 0.5 - orange, 1.0 - blue, 1.5 - magenta, and three values of the scale, $\sigma$, 0.25 - dashed, 0.5 - line, 1.0 - dash-dotted. For $\alpha=1$ (blue), a coupled exponential, $Z=\sigma$ and the coupled entropy is only slightly dependent on the coupling. For $\alpha=2$ (magenta), a coupled Gaussian, the normalization has a strong dependence on the coupling, which is amplified by the power 2; thus, taking the root $\gamma=\alpha$ may form a better metric. However, this same root amplifies the metric for $\alpha=\frac{1}{2}$ and low values of $\kappa$.} \label{fig_CE}
\end{figure*}
The coupled logarithm term is integrated over the independent equals distribution $F^{(\frac{1+(d+\alpha)\kappa}{1+d\kappa})}(\mathbf{x})$. The coupled-log of the partition function is a constant with respect to $\mathbf{x}$ and thus unchanged. Each $x_ix_j$ integral is equal to its scale parameter $\sigma_{ij}^\alpha$. Thus, these terms sum to $d^2$, so that the variable term reduces to $\frac{d}{\alpha}$, giving: 
\begin{align}
    \ln_{\alpha\kappa}\left(Z_\kappa^\frac{1}{1+d\kappa}\right)
    \oplus_{\alpha\kappa}\left(\frac{d}{\alpha}\right)
    &=\frac{d}{\alpha} + \left(1+\alpha\kappa\frac{d}{\alpha}\right)\frac{1}{\alpha\kappa}
    \left(Z_\kappa^\frac{\alpha\kappa}{1+d\kappa}-1\right)\\
    &=\frac{d}{\alpha} + \ln_\frac{\alpha\kappa}{1+d\kappa}Z_\kappa(\sigma,\alpha,d).
\end{align}
\end{proof}
\begin{remark}
\begin{enumerate}
    \item The uniqueness of the coupled entropy stems from the fact that its solution for its maximizing distributions is a) equal to the argument of the coupled stretched exponential distribution at $x-\mu=\sigma,$ and b) the definition of the CSED is unique in specifying a scale that is independent of the shape. In summary, the coupled entropy preserves for non-exponential distributions the property of the entropy for exponential distributions that the density of the distribution at the scale is the entropy of the distribution translated into the density via the corresponding generalized logarithm.
    \item In the limit, $\lim\limits_{\kappa\rightarrow\infty} \frac{\alpha\kappa}{1+d\kappa}=\frac{d}{\alpha}Z_\kappa(\sigma,\alpha,d)^\frac{\alpha}{d},$ which suggests that $\gamma=\frac{\alpha}{d}$ may be useful if an asymptotic solution proportional to $Z_\kappa$ is required.
    \item Figure \ref{fig_CE} shows the coupled entropy for the coupled stretched exponential with $d=1.$ The two figures contrast a) $\gamma=1$ and b) $\gamma=\alpha.$ 
    \item The coupled entropy is related to the Tsallis and Normalized Tsallis entropy by factors dependent on the coupling:
    \begin{equation}
        H_\kappa(\mathbf{p}) = \frac{H_{\kappa}^{\text{NT}}(\mathbf{p})}{1+d\kappa}
        = \frac{H_{\kappa}^{\text{T}}
        (\mathbf{p})}{(1+d\kappa)\sum_jp_j^{1+\frac{\alpha\kappa}{1+d\kappa}}}.
    \end{equation}
\end{enumerate}
\end{remark}

\subsection{Composability and Extensivity Properties of the Coupled Entropy}\label{subsecAxioms}

Having established the uniqueness of the coupled entropy based on the symmetry of its solution for the CSEDs, next, I prove two lemmas that specify the composability and extensivity properties of the coupled entropy. It is anticipated that these lemmas could be used to formulate a set of axioms that uniquely specify the coupled entropy. 

In the development of generalized entropies for complex systems, an assumption arose that the generalization of the additivity was directly connected to the requirements for extensive growth of the entropy. This assumption was a natural outcome of the symmetry between the divisive and exponential terms of the $q-$logarithm, $\frac{1}{1-q}(x^{1-q}-1).$ Lemma \ref{lem_CESD} established that this symmetry must be defined for the SF of the CESD in order to define the information scale, which precisely separates the location and asymptotic shapes. The derivatives necessary for the PDF necessarily break this symmetry, and thus, the inverse generalized logarithm required for coupled entropy also breaks that symmetry. Therefore, axioms for the coupled entropy must specify both composability and extensivity requirements. 

Tempesta \cite{tempesta_group_2011} defined and developed with colleagues \cite{rodriguez_new_2019, tempesta_universality_2020} the universal group entropies, which, in part, specify how extensivity requirements can be used to construct an entropy and its composition rule. For the coupled entropy, the composability is a trace component with a nonlinear coefficient of $\alpha\kappa$ and a non-trace component with a power and inverse power of $\gamma$. The extensivity is defined based on the probability of the states growing according to a coupled stretched exponential. This provides a model of both power-law and stretched exponential growth that connects directly with the coupling $\kappa$ and the two stretched exponential terms $\alpha$ and $\gamma$ of the coupled entropy function.

\begin{lemma}[Composability \lowercase{o}f the Coupled Entropy]
Given the coupled entropy as defined in  equation \eqref{equ_CE}, and two dependent subsystems of dimensions $d_{A,B}$ with coupled entropies, $H_\kappa(A;\alpha,\gamma,d_A)$ and $H_\kappa(B|A;\alpha,\gamma,d_B)$ then:
\begin{enumerate}
    \item for $\gamma=1$, the composability of subsystems has a nonlinear term with magnitude of $\alpha \kappa$, where $\alpha$ is the location shape parameter and $\kappa$ is the asymptotic shape and nonlinear coupling. The coupled entropy of the combined system $A \cup B$ is:
\begin{align}
    H_\kappa(A \cup B; \alpha,1,d_{A,B}) &= 
    H_\kappa(A;\alpha,1,d_A) + H_\kappa(B|A;\alpha,1,d_B) \\ \nonumber
    &+  \alpha\kappa H_\kappa(A;\alpha,1,d_A) H_\kappa(B|A;\alpha,1,d_B) \\
    &= H_\kappa(A;\alpha,1,d_A) \oplus_{\alpha\kappa} H_\kappa(B|A;\alpha,1,d_B);
\end{align}
\item for $\kappa=0$, the composability of subsystems requires the exponent $\gamma$. The coupled entropy of the combined system $A \cup B$ is:
    \begin{align} 
    H_\kappa(A \cup B; \alpha,1,d_{A,B}) = 
    \left(H_\kappa(A;\alpha,1,d_A)^\gamma + H_\kappa(B|A;\alpha,1,d_B)^\gamma\right)^\frac{1}{\gamma};
    \end{align}
    \item for $\gamma \neq 1$ and $\kappa \neq 0$, the composability of subsystems for the coupled entropy is

\end{enumerate}
        \begin{align}
            &H_\kappa(A \cup B; \alpha,\gamma,d_{A,B}) = \\ 
             &\left(
             H_\kappa(A;\alpha,\gamma,d_A)^\gamma + 
             H_\kappa(B|A;\alpha,\gamma,d_B)^\gamma + 
             \alpha\kappa H_\kappa(A;\alpha,\gamma,d_A)^\gamma 
            H_\kappa(B|A;\alpha,1,d_B)^\gamma 
            \right)^\frac{1}{\gamma}. \nonumber
        \end{align}
    
\end{lemma}
\begin{proof}
    First, note that $P(A\cup B)=P(A)P(B|A),$ thus for the composition of the entropy, it is sufficient to consider an independent system $P(C)=P(B|A),$ such that $P(A\cup C)=P(A)P(C).$ Since the summations over the states for each system are independent, it is the composition of the generalized logarithm that determines the composition of the generalized entropy. Finally, to highlight the distinction between the divisive and exponent terms of the coupled logarithm, the labels $s=\alpha\kappa$ and $t_{i,j}=\frac{1}{1+d_{i,j}\kappa}$ will be used.
    \begin{enumerate}
        \item The composition by the coupled sum, defined in the Methods equation \eqref{equ_csum}  with coefficient $s$, is shown to hold for the coupled logarithm, $\ln_s x,$ since
    \begin{align}
        \frac{1}{s}\left(\left(p_i^{t_i}q_j^{t_j}\right)^s-1\right) &=
        \frac{1}{s}\left(p_i^{st_i}-1\right) +
        \frac{1}{s}\left(q_j^{st_j}-1\right) +
        s\frac{1}{s}\left(p_i^{st_i}-1\right) 
        \frac{1}{s}\left(q_j^{st_j}-1\right). 
    \end{align}
Carrying through the arithmetic on the right, all of the single probabilities cancel, as does one of the $-\sfrac{1}{s}$ terms. Note that only $s$ is required to be equal for both systems, so strictly speaking $\alpha$ and $\kappa$ could be different; however, this would also constrain the relationship between the coupling and the dimensions. More realistically, this composition requires that $\alpha_i=\alpha_j$ and $\kappa_i=\kappa_j$ but allows $d_i \neq d_j.$
    \item The composition rule for $\gamma$ which is applied outside of the summations (or integrals) of the states, simply accounts for the summation of the entropy of independent systems.  So given $H(A)=H_{BGS}(A)^\frac{1}{\gamma},$ then  $H(A)^\gamma$ is equal to the BGS entropy and is thus additive for independent subsystems. The root $\sfrac{1}{\gamma}$ is then applied to the composed entropy. 
    \item The two compositional rules are combined. First, applying the power $\gamma$ to each subsystem, then applying the non-additive composition based on the coefficient $\alpha\kappa$, then applying the root $\sfrac{1}{\gamma}$ to the composed system.
    \end{enumerate}
\end{proof}
\begin{remark}
    The Suyari uniqueness axioms regarding the Tsallis entropy \cite{suyari_generalization_2004} include a statement regarding a free multiplicative term. What is now established is that the multiplicative term is not a free choice, which motivates consideration of the extensivity as an additional requirement with the composability. 
\end{remark}
\begin{lemma}[Extensivity \lowercase{of the} Coupled Entropy]
    Consider a $d-$dimensional system and require that the coupled entropy function utilizes a matching value of d.
    \begin{enumerate}
        \item For a system with a power-law number of elements, let the equiprobable states have a probability of
        \begin{align*}
            p_i^{-1}=W(N)&=\exp_\frac{1}{\rho}^{1+\frac{d}{\rho}}(\frac{N}{\sigma})\\
            &=\left(1+\frac{N}{\rho\sigma}\right)^{\rho+d}\\
            &\propto N^{\rho+d} \text{ for } N \gg\sigma.
        \end{align*} 
Then, the condition for extensive coupled entropy is $\kappa=\frac{1}{\rho}$ and $\gamma=\alpha=1$.  This condition is satisfied for all scales, $\sigma;$ however, for $N\gg\sigma$ the model simplifies to a pure power-law. 
        \item For a system dependent on a stretched exponential number of elements, let the equiprobable state have a  probability of
        \begin{align*}
            p_i^{-1}=W(N)=\exp{\left(\frac{1}{a}\left(\frac{N}{\sigma}\right)^{a}\right)} 
            \propto \exp\left(N^a\right).
        \end{align*} 
Then, the condition for extensive coupled entropy is $\gamma=a$ , $\kappa=0,$ while $\alpha$ and $d$ are  free parameters.
        \item For a system dependent on a coupled stretched exponential number of elements, let the equiprobable state have a  probability of
        \begin{align*}
            p_i^{-1}=W(N)=\exp_\frac{a}{\rho}^{1+\frac{d}{\rho}}\left(\frac{1}{a}\left(\frac{N}{\sigma}\right)^a\right)
            =\left(1+\frac{1}{\rho}\left(\frac{N}{\sigma}\right)^a\right)
            ^\frac{\rho+d}{a}.
        \end{align*}
Then, the condition for extensive coupled entropy is $\gamma=\alpha=a$ and $\kappa=\frac{1}{\rho}$.  For $N\gg\sigma,$ the growth of states is proportional to $N^{\rho+d}$ , and thus the extensivity condition reduces to $\kappa=\frac{1}{\rho}$. For $N\ll\sigma$, the growth of states is proportional to $\exp\left(N^a\right)$ , and the extensivity condition reduces to $\gamma=a$.
    \end{enumerate}
\end{lemma}
\begin{proof}
    \begin{enumerate}
        \item Given $W(N)=\exp_\frac{1}{\rho}^{1+\frac{d}{\rho}}(\frac{N}{\sigma})$ for the scaling of microstates, the growth of the coupled entropy is
        \begin{align}
            H_\kappa(p;\alpha,\gamma,d)&=\left(\ln_{\alpha\kappa}\left(\left(\exp_\frac{1}{\rho}^{1+\frac{d}{\rho}}(\frac{N}{\sigma})\right)^{\frac{1}{1+d\kappa}}\right)\right)^\frac{1}{\gamma} 
            =\left(\frac{1}{\alpha\kappa}\left(\left(1+\frac{N}{\rho\sigma}\right)^\frac{(\rho+d)\alpha\kappa}{1+d\kappa}-1\right)\right)^\frac{1}{\gamma} \\
        \end{align}
The inner exponent establishes a requirement that $\rho+d=\alpha(\frac{1}{\kappa}+d).$ Thus extensivity requires $\alpha=\gamma=1$ and $\rho=\frac{1}{\kappa}.$ The coupled entropy is then $H_\frac{1}{\rho}(p,1,1)=\frac{N}{\sigma}.$ The scale $\sigma,$ establishes a transition between exponential growth, $N\ll\sigma$ and power-law growth $N\gg\sigma$.
        \item Given $W(N)=\exp(\frac{1}{\alpha}\left(\frac{N}{\sigma}\right)^\alpha)$ for the scaling of microstates, the growth of the coupled entropy with $\kappa=0$ is 
        \begin{align}
            H_0(p;\alpha,\gamma)&=\left(\ln\exp \left(\frac{1}{a}\left(\frac{N}{\sigma}\right)^a\right)\right)^\frac{1}{\gamma} 
            =a^{-\frac{1}{\gamma}}+\left(\frac{N}{\sigma}\right)^\frac{a}{\gamma}.
        \end{align}
        Thus, extensive growth requires $\gamma=a$ but does not specify a requirement for $\alpha$ and $d.$
        \item Given $W(N)=\exp_\frac{1}{\rho}^{1+\frac{d}{\rho}}\left(\frac{1}{a}\left(\frac{N}{\sigma}\right)^a\right)$ for the scaling of the microstates, the growth of the coupled entropy is
        \begin{align*}
            H_\kappa(p;\alpha,\gamma,d)&=\left(\ln_{\alpha\kappa}\left(\left(
            \exp_\frac{1}{\rho}^{1+\frac{d}{\rho}}\left(\frac{1}{a}
            \left(\frac{N}{\sigma}\right)^a\right)\right)^
            {\frac{1}{1+d\kappa}}\right)\right)^\frac{1}{\gamma} \\
            &=\left(\frac{1}{\alpha\kappa}\left(
            \left(1+\frac{1}{\rho}\left(\frac{N}{\sigma}\right)^a\right)
            ^{\frac{\rho+d}{a}\frac{\alpha\kappa}{1+d\kappa}}-1\right)\right)^\frac{1}{\gamma}.
        \end{align*}
        Thus, for extensive growth, the internal exponent requires that $\alpha=a$ and $\rho=\frac{1}{\kappa}$. Simplifying the inner terms, $H_\kappa(p;\alpha,\gamma,d)=\left(\frac{1}{a}\left(\frac{N}{\sigma}\right)^a\right)^\frac{1}{\gamma}$, thus extensive growth requires that $\gamma=\alpha.$
    \end{enumerate}
\end{proof}
This completes the proof specifying the coupled entropy as a unique solution for the measurement of uncertainty in complex systems. Future research is recommended to formulate these proofs as a set of axioms that specify the coupled entropy; however, as explained earlier, it is the solution rather than a set of axioms that is critical. The discussion of applications will demonstrate why this solution is significant.

\section{Discussion of Applications}\label{secDiscuss}
Having established the uniqueness of the Coupled Entropy, a plethora of applications will benefit from the improved analytical precision of its solutions. I'll begin by comparing the coupling parameter with other proposed measures of statistical complexity. I'll then prove that the coupled thermostatistics defines the generalized temperature as the scale of the coupled exponential distribution and is thus independent of the nonlinearity. Next, I review the two applications of the coupled entropy framework to infodynamics, including the modeling of extreme heavy-tailed distributions with the coupled variational inference method and the analysis of communications systems with heavy-tailed interference.

\subsection{A Measure of Statistical Complexity}\label{subsecModels}
A consensus regarding a precise definition of statistical complexity has not been established, though there is agreement that purely ordered and disordered states do not have complex structure. There are two aspects of the coupling parameter that suggest it could be a measure of statistical complexity:
\begin{enumerate}
    \item The multiplicative noise model consists of an ordered function and two purely disordered sources of white noise. The addition of these components is still a non-complex system; however, the multiplicative noise component is mediated by the coupling $\kappa,$ and thus this coefficient could measure the degree of complexity.
    \item The BGS measure of information for the coupled exponential distribution is $1+\ln\sigma+\kappa.$ The non-complex information could be associated with the exponential distribution, $1+\ln\sigma$, thereby suggesting that $\kappa$ would be the complex information.
\end{enumerate}

To examine the role of the coupling in measuring complexity more systematically, let's consider a measure based on a comparison of the entropy of a complex distribution with its non-complex equilibrium distribution proposed by Shiner, Davison, and Landsberg (SDL) \cite{shiner_simple_1999}. The López-Ruiz, Mancini, and Calbet (LMC) \cite{ruiz_statistical_2013, martin_generalized_2006, rudnicki_monotone_2016} considers a broader range of divergences but reduces to the SDL function given the Kullback-Leibler divergence. These measures of complexity seek to account for the distance between order and disorder via the multiplication of complementary ratios with respect to the maximum or equilibrium entropy $H_e$,
\begin{align}
    C^{SDL}
    = \frac{H}{H_e}
    \left(1-\frac{H}{H_e}\right).
\end{align}
Take the exponential distribution and the Gaussian distribution to be the maximum entropy equilibrium distributions. If we consider the coupled exponential (ce) and coupled Gaussian (cg) distributions in which the scale is restricted to satisfy the respective constraint of the equilibrium (e) mean or standard deviation, then the scales of these distributions are $\sigma_{ce} = \sigma_{e}(1-\kappa)$ and $\sigma_{cg} = \sigma_{e}\sqrt{1-2\kappa}$, respectively. Note that this  restricts $\kappa<1$ and $\kappa<\sfrac{1}{2},$ respectively. For the coupled exponential case, the complexity measure is:
\begin{align}
    C_{ce}^{SDL}
    &= \left(
    \frac{1+\ln \sigma_{ce}+\kappa}{1+\ln \sigma_{e}}
    \right)
    \left(
    1 - \frac{1+\ln \sigma_{ce}+\kappa}{1+\ln \sigma_{e}}
    \right) \\
    &=\left(
    \frac{1+\ln (\sigma_{e}(1-\kappa))+\kappa}{1+\ln \sigma_{e}}
    \right)
    \left(
    1 - \frac{1+\ln (\sigma_{e}(1-\kappa))+\kappa}{1+\ln \sigma_{e}}
    \right) \\
    &\approx \frac{\kappa}{1+\ln \sigma_e} + O(\kappa^2)
\end{align}
The entropy of the coupled Gaussian to first order is $H_{cg}=\frac{1}{2}+\frac{1}{2}\ln 2\pi \sigma_{cg}+\kappa +O(\kappa^2)$. The complexity of the coupled Gaussian relative to the Gaussian is thus: 
\begin{align}
    C_{cg}^{SDL}&=\left(
        \frac{\frac{1}{2}+\frac{1}{2}\ln2\pi\sigma_e\sqrt{1-2\kappa}+\kappa}
        {\frac{1}{2}+\frac{1}{2}\ln2\pi\sigma_e}
    \right)
    \left(1-
        \frac{\frac{1}{2}+\frac{1}{2}\ln2\pi\sigma_e\sqrt{1-2\kappa}+\kappa}
        {\frac{1}{2}+\frac{1}{2}\ln2\pi\sigma_e}
    \right) \\
    &\approx \frac{\kappa}{1+\ln(2\pi\sigma_e)}+O(\kappa^2)
\end{align}

These computations confirm that, to first order, the nonlinear statistical coupling is a measure of the statistical complexity, with the refinement that dividing by a measure of the equilibrium entropy improves the metric. 

\subsection{Thermodynamics of Complex Systems}\label{subsec_thermo}

An essential test for the utility of a generalized entropy function is its ability to frame a comprehensive generalization of thermodynamics. The generalization of thermodynamics \cite{naudts_generalized_2004, naudts_generalised_2011} for complex systems is a the principal objective of the nonextensive statistical physics initiative. However, despite significant progress \cite{curado_general_1999, tsallis_introduction_2004, ferri_equivalence_2005,nobre_effective-temperature_2012} there remain difficulties in defining a consistent generalized temperature which supports a full thermodynamic framework.  I will proof that the coupled entropy enables two complementary definitions of temperature based on either the energy constraint or the inverse of the energy Lagrangian multiple. From these definitions consistent expressions for the free energy can be defined using either the coupled algebra \ref{subsec_CPA} or the standard algebra, respectively.

Consider a canonical ensemble with an infinite number of states and continuous energy levels $\epsilon,$ which is characterized by a NESS due to nonlinear drift-diffusion, superstatistics, or long-range correlations characterized by the degree of the nonlinear statistical coupling with either heavy-tailed $(\kappa>0)$ or compact-support $(-1<\kappa<0)$ distributions. For simplicity, the Boltzmann constant is set to one ${(k=1)}.$ Then, the generalization of the Boltzmann-Gibbs distribution is the coupled exponential distribution, and its associated coupled entropy is:
\begin{align}
    f(\epsilon)&=\frac{1}{\sigma}\exp_\kappa^{-(1+\kappa)} 
    \left(\frac{\epsilon}{\sigma}\right) \\
    S_\kappa(f(\epsilon)) & = 1 + \ln_\frac{\kappa}{1+\kappa} \sigma 
    = \frac{1+\kappa}{\kappa} \sigma^\frac{\kappa}{1+\kappa} - \frac{1}{\kappa}.
\end{align}
For this system, the internal energy of the system, $U_\kappa=\int_0^\infty \epsilon f/^\frac{1+2\kappa}{1+\kappa}(\epsilon)\mathrm{d}\epsilon = \sigma,$ and the partition function, $Z_\kappa= \int_0^\infty \exp_\kappa^{-(1+\kappa)} 
    \left(\frac{\epsilon}{\sigma}\right)\mathrm{d}\epsilon = \sigma,$ are both equal to the informational scale. With the Boltzmann-Gibbs equilibrium $(\kappa=0)$ the temperature is  equal to both this constraint and to the inverse of the associated Lagrangian multiplier. For the NESS, a consistent set of thermodynamic relationships is satifies with two complementary approaches.  First, the generalized temperature is defined by the energy constraint, $T_\kappa^C=\sigma$ and the thermodynamic equations utilize the coupled pseudo-algebra. Second, the nonlinearity is incorporated into the generalized temperature via the inverse of the Lagrange multiplier, $T_\kappa^M=\sigma^\frac{1}{1+\kappa},$ and the thermodynamic relationships utilize the standard alegbra.  The Lagrangian maximization of the coupled entropy and the value of its multipliers are proven in the Methods section \ref{sec_maxCE}. 

In equilibrium thermodynamics, the temperature is defined by the derivative of the entropy with respect to the energy. In the NESS system with $\kappa$ coupling this relationship defines the temperature to be equal to the Lagrange multiplier 
\begin{align}
    \frac{1}{T_\kappa^M} &\equiv \frac{\partial S_\kappa}{\partial U_\kappa} 
    \nonumber \\
    &= \frac{\partial (1 + \ln_\frac{\kappa}{1+\kappa} \sigma)}{\partial \sigma}
    \nonumber \\
    &= \sigma^{\frac{\kappa}{1+\kappa}-1} = \sigma^{-\frac{1}{1+\kappa}}.
\end{align}
I'll return to this model in a moment but of higher interest is the fact that by accounting for the role of the independent equals moment, a complimentary and consistent model emerges in which the generalized temperature is independent of the coupling. In this case, the derivative is with respect to a scaled value of the energy, $V_\kappa = \frac{1}{2+\kappa} U_\kappa^{2+\kappa},$ reflecting the modified relationship between entropy and energy for extensive growth.  The constraint based temperature is then,
\begin{align}
    \frac{1}{T_\kappa^C} &\equiv \frac{\partial S_\kappa}{\partial V_\kappa}
    = \frac{\partial S_\kappa}{\partial U_\kappa}
    \frac{\partial U_\kappa}{\partial V_\kappa}
    \nonumber \\
    &= \sigma^{-\frac{1}{1+\kappa}} \sigma^{1+\kappa} = \frac{1}{\sigma}.
\end{align}

In equilibrium, the free energy satisfies two expressions, $F=U-TS=-T\ln Z.$ Thus, the entropy can be expressed as $S= \frac{U}{T}+\ln Z.$ Given the uniqueness requirement \ref{def_required} regarding the structure of the coupled entropy solution, this suggests that the coupled sum and difference should be used for the generalized thermodynamic expressions. And indeed, this expression is consistent for $T_\kappa^C=\sigma,$ 
\begin{align}
    S_\kappa &\equiv \frac{U_\kappa}{T_\kappa^C}
    \oplus_\kappa\ln_\kappa Z^\frac{1}{1+\kappa}  \\
    &= 1 \oplus_\kappa\ln_\kappa \sigma^\frac{1}{1+\kappa} \\
    &= 1+ \ln_\frac{\kappa}{1+\kappa} \sigma \nonumber.
\end{align}
With this expression for the generalization of thermodynamics, the two free energy definitions are consistent,
\begin{align}
    F_\kappa &\equiv U_\kappa \ominus_\frac{\kappa}{T_\kappa^C} 
    T_\kappa S_\kappa 
    \equiv \ominus_\frac{\kappa}{T_\kappa^C} 
    T_\kappa^C \ln_\kappa Z^\frac{1}{1+\kappa}\\
    &= \sigma \ominus_\frac{\kappa}{\sigma} \sigma
    (1 \oplus_\kappa \ln_\kappa \sigma^\frac{1}{1+\kappa}) 
    \nonumber \\
    &=0 \ominus_\frac{\kappa}{\sigma} 
    \sigma \ln_\kappa \sigma^\frac{1}{1+\kappa} 
    \rightarrow \ominus_\frac{\kappa}{T_\kappa^C} 
    T_\kappa^C \ln_\kappa Z^\frac{1}{1+\kappa}. \checkmark
    \nonumber \\
\end{align}
This generalization of the free energy is consistent with treating the probabilities associated with the energy and the entropy of the system as independent, thereby requiring nonlinear combinations of their information.  That is, taking the coupled exponential of the free energy divided by temperature gives the following expression,
\begin{align}
    \exp_\kappa^{-(1+\kappa)} \frac{F_\kappa}{T_\kappa^C} 
    = \exp_\kappa^{-(1+\kappa)}
    \left(\frac{U_\kappa}{T_\kappa^C} \ominus_\kappa 
     S_\kappa \right) \nonumber \\
    = \frac{\exp_\kappa^{-(1+\kappa)}\frac{U_\kappa}{T_\kappa^C}}
    {\exp_\kappa^{-(1+\kappa)}S_\kappa}.
\end{align}

While the model with a generalized temperature independent of the nonlinearity is attractive, a complimentary model with the functions independent of the nonlinearity is also possible.  In this case, the moment temperature is used,
\begin{align}
    S_\kappa &= \frac{U_\kappa}{T_\kappa^M} + \ln_\kappa Z^\frac{1}{1+\kappa}
    \nonumber \\
    &= \sigma^{1-\frac{1}{1+\kappa}} 
    + \ln_\kappa \sigma^\frac{1}{1+\kappa} \nonumber \\
    &= \sigma^\frac{\kappa}{1+\kappa} 
    + \frac{1}{\kappa}(\sigma^\frac{\kappa}{1+\kappa} -1) \\
    &= \frac{1+\kappa}{\kappa}\sigma^\frac{\kappa}{1+\kappa} - \frac{1}{\kappa} \\
    &= 1 + \ln_\frac{\kappa}{1+\kappa} \sigma.
\end{align}
And the two expressions for the free energy are equal,
\begin{align}
    \text{First Expression:} \nonumber \\
    F_\kappa &= U_\kappa-T_\kappa S_\kappa \nonumber \\
    &= \sigma - \sigma^\frac{1}{1+\kappa}
    (1+\ln_\frac{\kappa}{1+\kappa} \sigma) \nonumber \\
    &= \sigma + \frac{1}{\kappa}\sigma^\frac{1}{1+\kappa}
    - \frac{1+\kappa}{\kappa} \sigma \nonumber \\
    &= \frac{1}{\kappa}(\sigma^\frac{1}{1+\kappa} - \sigma); \\
    \text{Second Expression:} \nonumber \\
    F_\kappa &= - T_\kappa \ln_\kappa Z_\kappa^\frac{1}{1+\kappa}
    \nonumber \\
    &= - \sigma^\frac{1}{1+\kappa} 
    \ln_\kappa \sigma^\frac{1}{1+\kappa} \nonumber \\
    &= \frac{1}{\kappa}(\sigma^\frac{1}{1+\kappa} - \sigma).
\end{align} 
So, the coupled entropy framework achieves two complementary, consistent models for defining a generalized temperature.

Physical evidence for the significance of this framework comes from the research of J. Cleymans, D. Worku, and colleagues \cite{cleymans_tsallis_2012,cleymans_systematic_2013} in their thermodynamic model of high-energy particle collisions. The momentum distribution of these collisions exhibits power-law behavior that can be accurately fitted to a generalized Pareto distribution over several log-scale decades. Although they cite the Tsallis $q-$statistics framework, they document the inability to formulate the thermodynamic equations using the $q-$exponential distribution. Instead, they correct the distribution by raising it to the power $q,$ which has the effect of redefining the variable, call it $q^{CM},$ and aligning the distribution with the requirements specified under the nonlinear statistical coupling framework, 
\begin{align}
    f(y) = \frac{1}{Z}\left(1 + (q^{CW}-1)\frac{g(y)}{T}\right)^\frac{q^{CW}}{1-q^{CW}}.
\end{align}
The particle physics details are not included in order to focus on the structure of the distribution. This distribution is now in the form originally specified by the generalized Pareto distribution, first recommended by this author in the preprint \cite{nelson_relationship_2008}, and proven here to be required to assure the temperature, $T,$ is equal to the informational scale. This can be seen via the following relationships for $\alpha=d=1,$
\begin{align}
    q^{CW}-1&=\frac{q-1}{2-q}=\kappa \\
    \frac{q^{CW}}{1-q^{CW}}&=\frac{1}{1-q}=-\left(\frac{1}{\kappa}+1\right).
\end{align}
Thus, the multiplicative factor and the exponent factor now have the required relation to each other. But of course, this modified distribution leads to a modified definition for the generalized logarithm and entropy. Therefore, the physical evidence from modeling high energy particle collisions provides support for the need to utilize the coupled entropy and the nonlinear statistical coupling framework to achieve consistent thermodynamic relationships. There is much additional research required in this area, since these relationships are just a first step in evaluating the hypothesis that the laws of thermodynamics can be generalized for complex systems.

\subsection{Complex Infodynamics}
A consistent generalization of thermodynamics is crucial to establishing a coherent model of infodynamics for complex systems, given the intimate connection between the two domains. In this section, I'll demonstrate how the coupled entropy has advanced the design of variational inference for extreme environments and I'll outline how similar impacts can be realized in the design of congested communication systems.

Within neuroscience, the predictive coding model \cite{friston_predictive_2009, smith_recent_2021} is a leading candidate for explaining the efficiency and effectiveness of biological intelligence. The model is based on local updates at each neuron in which error signals from lower layers are compared with predictions from higher layers, with the difference generating new error and prediction signals. Nevertheless, within artificial intelligence, global signaling via backpropagation \cite{millidge_predictive_2022} dominates the training of deep learning \cite{smith_recent_2021, rosenbaum_relationship_2022} algorithms. Variational inference \cite{blei_variational_2017} , in which probabilistic models are learned and utilized for generative and discriminative capabilities, may be a bridge between these two paradigms. Variational inference, including predictive coding, turns the intractable problem of learning an unknown posterior distribution into an optimization problem of approximating that distribution given a parameterized family of distributions. The negative of the evidence lower bound of the optimization is equivalent to the informational free energy and consists of a negative log-likelihood, which is the error signal between a generated and original dataset, and the divergence between the posterior and prior latent distribution, which is a regulator maintaining model simplicity.

Expanding the scope of variational inference to non-exponential distributions could be crucial to improving neurological models and accelerating the capability of artificial systems. Thus, variational inference methods that leverage non-exponential information are an active area of research \cite{goertzel_actpc-geom_2025}. Such methods could be deployed in the heavy-tailed domain to improve inference robustness \cite{cao_coupled_2022} or in the compact-support domain \cite{martins_nonextensive_2009} to improve efficiency via spareness. Efforts to utilize the Rényi \cite{li_renyi_2016}, Tsallis \cite{kobayashis_q-vae_2020}, and an initial implementation with the coupled entropy, were limited to the domain of finite variance $(\kappa<\sfrac{1}{2})$ to ensure convergence of the training. Recently, the Coupled Variational Autoencoder (CVAE) \cite{nelson_variational_2025} was designed to draw samples from the independent-equals modification to the latent coupled Gaussian distribution, thereby guaranteeing that even the most extreme heavy-tailed distributions with $\kappa \gg 1$ could be trained. This is because the coupling and scale of the latent distribution are transformed by $\frac{\kappa}{1+2\kappa}$ and $\frac{\sigma}{1+2\kappa},$ respectively. This design demonstrated improvements of $10\%$ and $25\%$ in measures of Learned Perceptual Image Patch Similarity and Multiscale Structural Similarity, respectively, across  coupling values ranging from $10^{-5} \text{ to } 10^5.$ The code base for the nonlinear statistical coupling methods and the CVAE algorithm is referenced in Appendix \ref{app_github}.

An open question about the design of the coupled variational inference algorithms is whether to use addition or coupled addition in combining the accuracy and regularization components. The reported results utilized addition; however, the thermodynamic computations clarified a requirement to utilize the coupled sum to combine the energy and temperature times entropy. Related to this question is whether the coupled divergence function should utilize a $f-$divergence style formulation in which the generalized logarithm is applied to the division of the probabilities or a Bregman style formulation in which the coupled logarithm of each probability is subtracted. The distinction is related to whether the probability distributions are independent, while the information is non-additive ($f-$divergence), or the distributions are dependent but their information is additive (Bregman). Because of the subtleties of these distinctions, the information geometric analysis of the coupled divergences is deferred for a separate investigation \cite{amari_information_2016,nielsen_elementary_2020}.

In parallel with the complexity of AI systems, modern communications systems, such as 5G wireless and cognitive radio networks, are facing increased complexity as dense traffic creates heavy-tailed interference. In contrast to a Gaussian noise channel, interference fluctuates in intensity depending on the instantaneous use of the channel. For example, Clavier, et al. measured \cite{clavier_experimental_2021} tails shapes above $\kappa>1,$ in 5G networks, indicative of quite extreme fluctuations. In cognitive radio designs \cite{munoz_renyi_2020} secondary users much adaptively switch to open channels in order to guarantee non-interference with primary users. These applications cannot use standard metrics such as the mean-square error or the Shannon entropy, since the outliers invalidate the metrics foundational assumptions.

The Rényi entropy has found wide use in such applications \cite{principe_information_2010, hild_analysis_2006, erdogmus_adaptive_2004}. The probability analysis of entropy function presented in Section \ref{subsec_required} provides insight regarding why the generalized entropies are effective in these circumstances. A high-interference channel characterized by $1<\kappa<2, $ such as the distributions shown in Figure \ref{fig_entropies}, has a Shannon entropy measure dominated by the tail. The Rényi entropy dramatically reduces this influence, if its index matches this tail shape via $q=1+\frac{\alpha\kappa}{1+d\kappa}.$\footnote{As explained in Appendix \ref{app_Entropies} $q$ is used for the Rényi index, since $\alpha$ is being used for the stretching parameter.} The coupled entropy framework is expected to further improve communication design in the presence of heavy-tailed noise in two ways. First, in the Preliminary section \ref{subsec_noise} the connection between multiplicative noise, heavy-tailed tail shape, and the index of the coupled entropy were made explicit, which will make the design choices more explicit. Second, generalized entropy requirement reviewed in the Main Results section \ref{subsecEntropies} show that while the Rényi partially removes the influence of the tail shape, the full discounting of this influence requires the coupled entropy.

\section{Method: Maximization of the Coupled Entropy}\label{sec_maxCE}
A core methodology of infodynamics is the derivation of the distribution that maximizes an entropy function using the Lagrangian function. In this section, I prove that the coupled stretched exponential distributions maximize the coupled entropy. A foundational weakness in the development of the Tsallis entropy was the assumption that given a hypothesized entropy, the maximizing distribution is derived and thus to be accepted without validation. As I established in the \ref{subsecEntropies}, in fact, the definition of the shape-scale distributions was established correctly in the 1800s, and the \textit{q-}exponentials have a weaker justification. Thus, a stronger requirement for a generalized entropy is that a) the shape parameter and thus the nonlinearity or complexity defines the generalization, and b) that the maximizing distributions are the shape-scale distributions given the independent-equals constraints. 

The derivation for one dimension is completed here, though the multivariate distributions is also recommended for future research. An additional simplification is the use of just two constraints, the normalization and the $\alpha$ moment with $\mu=0$. The coupled entropy is defined using the independent equals distribution with the power $\frac{\alpha\kappa}{1+\kappa},$ which explicitly facilitates the computation of the $\alpha$ moment. The computation and thus the maximization of the coupled entropy with multiple constraints and statistics is an important investigation to complete. This is recommended as a part of a broader research program on the information geometry of the coupled exponential family.

\begin{lemma}[Coupled Entropy Maximized by the Coupled Exponential Distribution]
    Given the coupled entropy 
    \begin{equation} H_\kappa(f(x)) = \int_{x\in\mathcal{X}}  
        f/^{\left(\frac{1+(1+\alpha)\kappa}{1+\kappa}\right)}(x)\ln_{\alpha\kappa}\left(f(x)\right)^{-\frac{1}{1+\kappa}}
    \end{equation} 
and  two constraints, the normalization and  the independent equals moment, $\mu_1^{(1+\frac{\alpha\kappa}{1+\kappa})}=\sigma^\alpha$;
then, the maximum coupled entropy distribution is the coupled stretched exponential of equation \eqref{equ_csePDF}, 
\begin{align}
f_\kappa(x;\alpha) &= \frac{1}{Z_\kappa(\sigma,\alpha)}\left(1+\kappa\frac{x^\alpha}{\sigma^\alpha}\right)^{-\frac{1+\kappa}{\alpha\kappa}}
=\exp_{\alpha\kappa}^{-(1+\kappa)}
\left(\frac{x^\alpha}{\alpha \sigma^\alpha}
\oplus_{\alpha\kappa} \ln_{\alpha\kappa} Z^\frac{1}{1+\kappa}
\right).
\end{align}
The Lagrangian multiples are
\begin{align}
    \begin{matrix}
        \text{Normalization:} & \lambda_0 =  
        - \frac{1} {(1+\kappa) 
    \int_{x\in\mathcal{X}} 
    f(x)^{1+\frac{\alpha\kappa}{1+\kappa}} dF(x)}; \\
    \text{Indpendent Equals Moment:} & \lambda_1 =
    \frac{Z_\kappa^\frac{\alpha\kappa}{1+\kappa} }
     {\alpha \sigma^\alpha}.
    \end{matrix}
\end{align}
\end{lemma}

\begin{proof}
    Given the coupled entropy function and the two constraints, the Lagrangian function is:
\begin{align}
\mathcal{L} &= \int_{x\in\mathcal{X}} f/^{\left(1+\frac{\alpha\kappa}{1+\kappa}\right)} (x)
\ln_{\alpha\kappa} f(x)^{-\frac{1}{1+\kappa}} dF(x) 
+ \lambda_0 \left(1 - \int_{x\in\mathcal{X}} f(x)  dF(x) \right) \\
&\ \ \ + \lambda_1 \left( \sigma^\alpha -
\int_{x\in\mathcal{X}} x^\alpha f/^{\left(1+\frac{\alpha\kappa}{1+\kappa}\right)} (x)dF(x)
\right),
\end{align}
where $\lambda_0$ and $\lambda_1$ are the Lagrangian multiples for the normalization and independent equals alpha moment, respectively. For maximization, the derivative must be zero, \(\delta \mathcal{L} / \delta f = 0\). The derivative is in terms of a particular value $y=x'$ such that the integrals only have a non-zero derivative at $y$:
\begin{enumerate}
\item Entropy Derivative
            \begin{align}
            &\text{Separate the integrals within the coupled entropy:} \nonumber \\
            H_\kappa(\textbf{X})=\frac{A}{B}; 
            A &= \int_{x\in\mathcal{X}} f(x)
            ^{1+\frac{\alpha\kappa}{1+\kappa}}
\ln_{\alpha\kappa} f(x)^{-\frac{1}{1+\kappa}} dF(x); \ 
            B = \int_{x\in\mathcal{X}} f(x)
            ^{1+\frac{\alpha\kappa}{1+\kappa}} dF(x) \nonumber \\
                \frac{\delta H_\kappa(f(y))}{\delta f(y)}
                &= \frac{\frac{1+(1+\alpha)\kappa}{1+\kappa}}{B}
                f(y)^\frac{\alpha\kappa}{1+\kappa} 
                \ln_{\alpha\kappa} f(y)^{-\frac{1}{1+\kappa}}
                -\frac{1}{B(1+\kappa)} \nonumber \\
    &- \frac{1+(1+\alpha)\kappa}{1+\kappa}
    f(y)^\frac{\alpha\kappa}{1+\kappa} \frac{A}{B^2}.
            \end{align} 
\item The normalization derivative is $-\lambda_0$.
\item Independent equals derivative
            \begin{align}
                &\frac{\delta}{\delta f(y)}  \lambda_1 \left( \sigma^\alpha 
                - \int_{x\in\mathcal{X}} x^\alpha
                f/^{\left(1+\frac{\alpha\kappa}{1+\kappa}\right)} 
                (x)dF(x) \right) \nonumber \\
                &= -\lambda_1 y^\alpha \frac{1+(1+\alpha)\kappa}{1+\kappa}
                \frac{f(y)^\frac{\alpha\kappa}{1+\kappa}}{B} 
                \nonumber \\
                &+\lambda_1 \frac{1+(1+\alpha)\kappa}{1+\kappa}
                f(y)^\frac{\alpha\kappa}{1+\kappa}
                \frac{\int_{x\in\mathcal{X}} x^\alpha
                f(x)^{1+\frac{\alpha\kappa}{1+\kappa}}dF(x)}
                {B^2} \nonumber \\
                &= -\lambda_1\frac{1+(1+\alpha)\kappa}{(1+\kappa)B}
                f(y)^\frac{\alpha\kappa}{1+\kappa}
                \left( y^\alpha - \sigma^\alpha\right)
            \end{align}
\end{enumerate}
Combining terms, the Lagrangian derivative is:
\begin{align}\label{equ_Lder}
    \delta \frac{\mathcal{L}}{\delta f} 
    &= \frac{1+(1+\alpha)\kappa}{(1+\kappa)B}
    f(y)^\frac{\alpha\kappa}{1+\kappa}
    \left(\ln_{\alpha\kappa} f(y)^{-\frac{1}{1+\kappa}}
    -\frac{A}{B}-\lambda_1 (y^\alpha - \sigma^\alpha)\right) \nonumber \\
     &-\frac{1}{B(1+\kappa)} - \lambda_0 \nonumber \\
    &=0.
\end{align}
Multiplying by $-(1+\kappa)\alpha\kappa B$ and solving for $f(y)$ gives:
\begin{align}
    &(1+(1+\alpha)\kappa)
    \left(1 + \frac{A}{B}\alpha\kappa
    +\lambda_1 \alpha\kappa(y^\alpha - \sigma^\alpha) \right)
    f(y)^\frac{\alpha\kappa}{1+\kappa} \nonumber \\
    &= (1+(1+\alpha)\kappa)
     -\alpha\kappa - \lambda_0 (1+\kappa)\alpha\kappa B \nonumber \\
     f(y)^{-\frac{\alpha\kappa}{1+\kappa}} 
     &= \frac{(1+(1+\alpha)\kappa)
    \left(1 + \frac{A}{B}\alpha\kappa
    +\lambda_1 \alpha\kappa(y^\alpha - \sigma^\alpha) \right)}
     {(1+(1+\alpha)\kappa)
     -\alpha\kappa - \lambda_0 (1+\kappa)\alpha\kappa B} \nonumber \\
     &= \frac{(1+(1+\alpha)\kappa)
     \left(1 + \frac{A}{B}\alpha\kappa
    -\lambda_1 \alpha\kappa \sigma^\alpha\right)
    \left(1 + \frac{\lambda_1 \alpha\kappa y^\alpha}
    {1 + \frac{A}{B}\alpha\kappa 
    -\lambda_1 \alpha\kappa \sigma^\alpha} 
    \right)}
     {(1+(1+\alpha)\kappa)
     -\alpha\kappa - \lambda_0 (1+\kappa)\alpha\kappa B}
\end{align}
The right-hand side has the form $a(1+b y^\alpha),$ confirming that the coupled exponential distribution is the maximizing distribution. The Langragian multipliers are determined from the constraints, which specify that $a=Z^\frac{\alpha\kappa}{1+\kappa}$ and $b=\frac{\kappa}{\alpha\sigma^\alpha}.$ Given $ H_\kappa(f(y)) = \frac{A}{B} = \frac{1}{\alpha}+\ln_\frac{\alpha\kappa}{1+\kappa}Z_\kappa$ from Lemma \eqref{lem_unique} 
the moment constraint multiple is determined from the equation:
\begin{align*}
    \frac{1}{\sigma^\alpha} 
    &= \frac{\lambda_1 \alpha}
    {1 + \left(\frac{1}{\alpha}+\ln_\frac{\alpha\kappa}{1+\kappa}Z_\kappa\right)\alpha\kappa 
    -\lambda_1 \alpha\kappa \sigma^\alpha} \\
    &= \frac{\lambda_1 \alpha}
    {(1+\kappa) Z_\kappa^\frac{\alpha\kappa}{1+\kappa} 
    -\lambda_1 \alpha\kappa \sigma^\alpha} \\
     \frac{(1+\kappa) Z_\kappa^\frac{\alpha\kappa}{1+\kappa} }{\sigma^\alpha} 
     &= \lambda_1 \alpha(1 + \kappa ) \\
     \lambda_1 &= \frac{
     Z_\kappa^\frac{\alpha\kappa}{1+\kappa} }
     {\alpha \sigma^\alpha}.
\end{align*}
The normalization constraint is determined by the expression for the normalization:
\begin{align*}
    Z_\kappa^{-\frac{\alpha\kappa}{1+\kappa}} &= \frac{(1+(1+\alpha)\kappa)
     -\alpha\kappa - \lambda_0 (1+\kappa)\alpha\kappa B}
     {(1+(1+\alpha)\kappa)
     \left((1+\kappa) Z_\kappa^\frac{\alpha\kappa}{1+\kappa} 
    -Z_\kappa^\frac{\alpha\kappa}{1+\kappa} \kappa \right)} \\  
    \lambda_0 &=  
    - \frac{1} {(1+\kappa) 
    \int_{x\in\mathcal{X}} 
    f(x)^{1+\frac{\alpha\kappa}{1+\kappa}} dF(x)}
\end{align*}
Thus, completing the proof.
\end{proof}
\begin{remark}
    The Lagrangian multiplier, $\lambda_1,$ is the natural parameter of information geometry, dual to the moment parameter  $\sigma^\alpha.$ The presence in the Lagrangian multipliers of terms involving the integral of the distribution, $B=\int_\mathbf{X}f(x)^q dF(\mathbf{x}),$ and $Z_\kappa^\frac{\alpha\kappa}{1+\kappa}$ is related to the nonlinear deformation intrinsic to the coupled exponential family. As noted in the previous section, use of the partition function to normalize the distribution, and the coupled sum to bring this factor within the coupled exponential function, contrasts with the definitions used by Amari to define the $q-$exponential family. Thus, a complete study of the information geometry of the coupled exponential family is merited.
\end{remark}
\begin{remark}
    In the continuous case with $\alpha=1,$ the normalization of the coupled exponential distribution is $Z_\kappa = \sigma.$ In this case, the Lagrangian multipliers are $\lambda_0 = \sigma^\frac{\kappa}{1+\kappa}$ and $\lambda_1 = \sigma^{\frac{\kappa}{1+\kappa}-1} = \sigma^\frac{-1}{1+\kappa}.$ 
\end{remark}
\section{Conclusion}\label{sec_concl}
In this paper, I have established a set of proofs regarding the coupled entropy as the unique function that satisfies the requirements for an infodynamic measure for complex systems. This result advances the entropic modeling of complex infodynamics initiated by Rényi, who originated the use of the generalized mean, and Tsallis, who showed that the constraints for entropy maximization should raise the probability distribution to a power. Two core proofs show that the Tsallis $q-$statistics utilized a flawed definition for scale, which is corrected by the coupled entropy:
\begin{enumerate}
    \item The information scale is defined by the derivative of the log-log distribution equaling $-1$, and is the only scale that is independent of the tail shape and nonlinear source of uncertainty.
    \item The coupled entropy of its maximizing coupled exponential family distribution is unique in its isolation of the nonlinear statistical coupling to a generalized logarithm of the partition function. Further, this solution is equal to the argument of the coupled stretched exponential distribution at the location plus the scale.
\end{enumerate}

    Significantly, a generalization of thermodynamics for complex systems can now be grounded in a consistent set of relationships. A novel result reported here is two complementary but internally consistent generalizations of the foundational thermodynamic relationships.  The first and likely the preferred model has a generalized temperature equal to the scale constraint, $T_\kappa^C=\sigma.$ This approach requires that the free energy combine the internal energy and the entropy via the coupled subtraction, 
\begin{equation}
    F_\kappa^C \equiv U_\kappa \ominus_\frac{\kappa}{T_\kappa^C} T_\kappa S_\kappa = \ominus_\frac{\kappa}{T_\kappa^C} T_\kappa^C \ln_\kappa Z^\frac{1}{1+\kappa}.
\end{equation} 
The second model has a generalized temperature equal to the inverse of the Lagrange multiplier for the moment constraint, $T_\kappa^M =\sigma^\frac{1}{1+\kappa}.$ In this case, the free energy is generalized using the standard subtraction function:
\begin{equation}
    F_\kappa^M \equiv U_\kappa - T_\kappa S_\kappa = - T_\kappa^C \ln_\kappa Z^\frac{1}{1+\kappa}.
\end{equation} 
These two models show that the effects of nonlinearity must be included within the structure of the generalized thermodynamics, though there is a choice whether that is done within the measure of temperature or free energy. An important research question is whether these two models are related to the the dual coordinates of information geometry. 

Another important research investigation to complete is a thorough evaluation of the stability of the coupled entropy. Because there are open questions regarding the best way to pursue such an investigation, this is best addressed separately from this proof of the coupled entropy's uniqueness. For $\alpha=1$ and $d=1$, the finite coupled entropy of the coupled exponential distribution even as $\kappa \rightarrow \infty$ is strong evidence of its stability. As $\alpha$ and $d$ are varied stability requirements may impose restrictions on the domain of $\kappa$; however, it's possible that the use of $\gamma=\frac{\alpha}{d}$ would alleviate these restrictions. 

The initial performance of the coupled VAE \cite{nelson_variational_2025}, which achieves an original result enabling the modeling of extreme heavy-tailed distributions, is suggestive that applications of the coupled entropy are wide and significant. The design of systems that can respond adaptively to the complexities of a natural environment is fundamental to modern science and engineering. The heavy-tailed interference of dense, adaptive communications networks is indicative of these challenges. AI developers aspire to mimic and integrate with biological intelligence \cite{ikle_artificial_2025}; city planners seek ways of assuring that human habitats can flourish from and contribute back to natural habitats \cite{white_modeling_2015}; and modern finance requires management of complex risks impacted by political and environmental instabilities \cite{glenn_global_2025}. Each of these domains requires detailed modeling of the nonlinearities that create fluctuating noise patterns. The coupled entropy will be a crucial tool in the development of precise analytical tools for the design and analysis of complex systems.

\section*{Acknowledgements}

I wish to thank colleagues Amenah Al-Najafi, Igor Oliveira, William Thistleton, and Calden Wloka, whose co-development of the coupled variational inference methods clarified the unique importance of the coupled entropy. Special thanks go to Ugur Tirnakli and Bruce Boghosian for the invitation to the 2025 Nonextensive Statistical Physics Workshop to present the research on the coupled entropy.

\section*{Declarations}

There are no conflicts of interest or funding sources to declare.

DeepSeek and Mathematica were used as analytical aides in the research and development process; however, the manuscript was written manually, including the proofs and equations. Wikipedia was used for background information, though references refer to primary sources.

\appendix

\section{Generalized entropy functions}\label{app_Entropies}

This appendix will show the various forms in which the generalized entropies can be expressed. Of particular interest is deriving the form in which a generalized mean aggregates the probabilities and then is transformed into an entropy measure via a generalized logarithm. This will assist in explaining the relationship between the entropies. The expressions will utilize the relationship between $q$ and the coupled algebra via $q=1+\frac{\alpha\kappa}{1+d\kappa}=\frac{1+(d+\alpha)\kappa}{1+d\kappa}.$ Within the integrals, $f(\mathbf{x})$ or 
\begin{align}
    f/^q(\mathbf{x})
    =\frac{f(\mathbf{x})^q}
    {\sum_{\mathbf{x\in\mathcal{X}}} f(\mathbf{x})^{\left(q\right)}}
\end{align} 
will be separated out to clarify the difference between the probability serving as a weight and the function of the probabilities being averaged. The generalized mean associated with each entropy can be shown by applying the inverse of its respective generalized logarithm. 

\textbf{Rényi Entropy}
\begin{align}
    H_q^R(\mathbf{X}) &= -\ln\left(\int_{\mathbf{x\in\mathcal{X}}} f(\mathbf{x})f(\mathbf{x})^{q-1}dF(\mathbf{x})\right )^\frac{1}{q-1} \\
&=-\ln\left(\int_{\mathbf{x\in\mathcal{X}}}f(\mathbf{x})f(\mathbf{x})^
\frac{\alpha\kappa}{1+d\kappa}dF(\mathbf{x})\right)^
{\frac{1+d\kappa}{\alpha\kappa}}
\end{align}

\textbf{Tsallis Entropy}
\begin{align} \label{equ_TEnt}
     H_q^T(\mathbf{X}) &=\frac{1}{1-q}
     \left(\int_{
     \mathbf{x\in\mathcal{X}}}
     f(\mathbf{x})f(\mathbf{x})^{q-1}dF(\mathbf{x)}
     -1\right) \\
     H_\kappa^{T}(\mathbf{X};\alpha,d)
     &= {-\frac{1+d\kappa}{\alpha\kappa}}
     \left(\int_{\mathbf{x\in\mathcal{X}}}
     f(\mathbf{x})f(\mathbf{x})^
     \frac{\alpha\kappa}{1+d\kappa}
     dF(\mathbf{x})
     -1\right) \\
     &=- \int_{\mathbf{x\in\mathcal{X}}} f(\mathbf{x})
     \ln_\frac{\alpha\kappa}{1+d\kappa} f(\mathbf{x})
     dF(\mathbf{x})
\end{align}

\begin{align}
    \exp_\frac{\alpha\kappa}{1+d\kappa}(- H_\kappa^{T}(\mathbf{X})
    &= \exp_\frac{\alpha\kappa}{1+d\kappa}\left( \int_{\mathbf{x\in\mathcal{X}}} f(\mathbf{x})
     \ln_\frac{\alpha\kappa}{1+d\kappa} f(\mathbf{x})
     dF(\mathbf{x})
     \right) \\
     &= \left(1+\int_{\mathbf{x\in\mathcal{X}}} 
     \left(f(\mathbf{x})f(\mathbf{x})^\frac{\alpha\kappa}{1+d\kappa}
     -f(\mathbf{x})\right)dF(\mathbf{x})\right)^\frac{1+d\kappa}{\alpha\kappa} \\
     &= \left(\sum_{\mathbf{x\in\mathcal{X}}} 
     f(\mathbf{x})f(\mathbf{x})^\frac{\alpha\kappa}{1+d\kappa}
     dF(\mathbf{x})\right)^\frac{1+d\kappa}{\alpha\kappa}
\end{align}

\textbf{Normalized Tsallis Entropy}
\begin{align}
     H_q^{NT}(\mathbf{X}) &=\frac{1}{1-q}
     \left(1-
     \frac{1}{\int_{\mathbf{x\in\mathcal{X}}}
     f(\mathbf{x})^q \ dF(\mathbf{x})}\right) \\
     &=\frac{1}{1-q}\int_{\mathbf{x}\in\mathcal{X}}
     f/^q(\mathbf{x})\left(1-f(\mathbf{x})^{-q}\right)dF(\mathbf{x}) \\
     H_\kappa^{NT}(\mathbf{X};\alpha,d)
     &={\frac{1+d\kappa}{\alpha\kappa}}
     \int_{\mathbf{x\in\mathcal{X}}}
     f/^\frac{1+(d+\alpha)\kappa}{1+d\kappa}(\mathbf{x})
     \left(f(\mathbf{x})^{-\frac{1+(d+\alpha)\kappa}{1+d\kappa}}-1\right)
     dF(\mathbf{x})\\
     &=\frac{1+(d+\alpha)\kappa}{\alpha\kappa}
     \int_{\mathbf{x\in\mathcal{X}}}
     f/^\frac{1+(d+\alpha)\kappa}{1+d\kappa}(\mathbf{x})
     \ln_\frac{1+(d+\alpha)\kappa}{1+d\kappa} f(\mathbf{x})^{-1}
     dF(\mathbf{x})\\
     &=\int_{\mathbf{x\in\mathcal{X}}}
     f/^\frac{1+(d+\alpha)\kappa}{1+d\kappa}(\mathbf{x})
     \ln_\frac{\alpha\kappa}{1+d\kappa} 
     \left(f(\mathbf{x})^
     {-\frac{1+(d+\alpha)\kappa}{\alpha\kappa}}\right)
     dF(\mathbf{x})
\end{align}

\begin{align}
      &\left(\exp_\frac{\alpha\kappa}{1+d\kappa} 
      \left(H_\kappa^{NT}(\mathbf{X})\right)\right)^
      {-\frac{\alpha\kappa}{1+(d+\alpha)\kappa}} \\
      &= \left(\exp_\frac{\alpha\kappa}{1+d\kappa} 
      \left(
      \int_{\mathbf{x\in\mathcal{X}}}
     f/^\frac{1+(d+\alpha)\kappa}{1+d\kappa}(\mathbf{x})
     \ln_\frac{\alpha\kappa}{1+d\kappa} f(\mathbf{x})^
     {-\frac{1+(d+\alpha)\kappa}{\alpha\kappa}} 
      dF(\mathbf{x})\right)\right)^
      {-\frac{\alpha\kappa}{1+(d+\alpha)\kappa}} \\
      &= \left(1+
      \int_{\mathbf{x\in\mathcal{X}}}
     \left(f/^\frac{1+(d+\alpha)\kappa}{1+d\kappa}(\mathbf{x})
      f(\mathbf{x})^
     {-\frac{1+(d+\alpha)\kappa}{1+d\kappa}} 
     - f/^\frac{1+(d+\alpha)\kappa}{1+d\kappa}(\mathbf{x}) \right)
      dF(\mathbf{x})\right)^
      {-\frac{1+d\kappa}{1+(d+\alpha)\kappa}} \\
      &= \left(
      \int_{\mathbf{x\in\mathcal{X}}}
      f/^\frac{1+(d+\alpha)\kappa}{1+d\kappa}(\mathbf{x})
      f(\mathbf{x})^{-\frac{1+(d+\alpha)\kappa}{1+d\kappa}} 
      dF(\mathbf{x})\right)^
      {-\frac{1+d\kappa}{1+(d+\alpha)\kappa}} 
\end{align}
\newpage
\textbf{Coupled Entropy}
\begin{align}
    H_\kappa(\mathbf{X};\alpha,d,\gamma=1) &=\frac{1}{\alpha\kappa}
    \int_{\mathbf{x\in\mathcal{X}}} 
    f/^\frac{1+(d+\alpha)\kappa}{1+d\kappa}(\mathbf{x})
    \left(f(\mathbf{x})^{-\frac{\alpha\kappa}{1+d\kappa}}
    -1\right)dF(\mathbf{x}) \\
    &=\int_{\mathbf{x\in\mathcal{X}}} 
    f/^\frac{1+(d+\alpha)\kappa}{1+d\kappa}(\mathbf{x})
    \ln_{\alpha\kappa} f(\mathbf{x})^{-\frac{1}{1+d\kappa}}
    dF(\mathbf{x})
\end{align}
\begin{align}
    &\exp_{\alpha\kappa}^{-(1+d\kappa)} [H_\kappa(\mathbf{X})]  \\
    &=\exp_{\alpha\kappa}^
    {-(1+d\kappa)}\left(
    \int_{\mathbf{x\in\mathcal{X}}} 
    f/^\frac{1+(d+\alpha)\kappa}{1+d\kappa}(\mathbf{x})
    \ln_{\alpha\kappa} f(\mathbf{x})^{-\frac{1}{1+d\kappa}}
    dF(\mathbf{x})\right) \\
    &= \left(1+\frac{\alpha\kappa}{\alpha\kappa}\int_{\mathbf{x\in\mathcal{X}}}
    \left(f/^\frac{1+(d+\alpha)\kappa}{1+d\kappa}(\mathbf{x})
    f(\mathbf{x})^{-\frac{\alpha\kappa}{1+d\kappa}}
    -f/^\frac{1+(d+\alpha)\kappa}{1+d\kappa}(\mathbf{x})\right)
    dF(\mathbf{x})\right)^{-\frac{1+d\kappa}{\alpha\kappa}} \\
    &= \left(\int_{\mathbf{x\in\mathcal{X}}}
    f/^\frac{1+(d+\alpha)\kappa}{1+d\kappa}(\mathbf{x})
    f(\mathbf{x})^{-\frac{\alpha\kappa}{1+d\kappa}}
    dF(\mathbf{x})\right)^{-\frac{1+d\kappa}{\alpha\kappa}}
\end{align}
In Section \ref{subsec_required}, the entropies for the centered coupled exponential distribution $(\alpha=1,d=1,\mu=0)$ are mapped onto the distribution, which provides a visual comparison of the generalized entropies. The density value equals $\exp_\kappa^{-(1+\kappa)}(H).$ The variable is then determined by solving for $x,$ given a coupled exponential, 
\begin{align}
    x&=\sigma (H\ominus_\kappa \ln_\kappa \sigma^\frac{1}{1+\kappa})
    =\sigma\frac{H-\ln_\kappa \sigma^\frac{1}{1+\kappa}}{1+\kappa 
    \ln_\kappa \sigma^\frac{1}{1+\kappa}}.
\end{align} 

\section{Hanel-Thurner classification of the coupled entropy} \label{app_Hanel}

Hanel and Thurner \cite{hanel_comprehensive_2011, hanel_how_2014, amigo_brief_2018, korbel_classification_2018} derived a classification of generalized entropies for complex systems based on the scaling of the microstates, W. The classification is based on a scaling by $\lambda W,$ which has a limit of $W^{1-c},$ and a scaling by $W^{1+a},$ which has a scaling of $(1+a)^d.$ $c$ is associated with deformation of the exponential and is equal to $0<q<1$ for the Tsallis entropy. $d$ is a secondary scaling that describes the rate of convergence to the power-law. In this classification, generalized entropies are characterized by their trace and non-trace components, with $p_i=W^{-1}:$
\begin{align}
    F_{G,g}(\mathbf{p})=G\left(\sum_i^W g(p_i)\right).
\end{align}

For equiprobable states, the independent equals distribution is equal to the original distribution. Using the radial variable r and considering $d-$dimensions, the functions for  the coupled entropy are $g(r)=r \ln_{\alpha\kappa} r^{-\frac{1}{1+d\kappa}}$$=\frac{1}{\alpha\kappa}\left(r^{1-\frac{\alpha\kappa}{1+d\kappa}}-r\right),$ and $G(u)=u^\frac{1}{\gamma}.$ The scaling is approximated by
\begin{align}
    F_{G,g}(\mathbf{p}_{\lambda W}) &\approx G\left(\lambda W 
    g\left(\frac{1}{\lambda W}\right)\right) \\
    F_{G,g}(\mathbf{p}_{W^{1+a}}) &\approx G\left(W^{1+a}g\left(\frac{1}{W^{1+a}}\right)\right).
\end{align}
The $(c,d)$ classification of an entropy function is determined by the limit of the ratios
\begin{align}
    \lim_{W\rightarrow\infty}\left[
    \frac{G\left(\lambda W 
    g\left(\frac{1}{\lambda W}\right)\right)}{G\left(W 
    g\left(\frac{1}{W}\right)\right)}\right] 
    = \lambda^{1-c} \\
    \lim_{W\rightarrow\infty}\left[
    \frac{G\left(W^{1+a} 
    g\left(\frac{1}{W^{1+a}}\right)\right)}{G\left(W 
    g\left(\frac{1}{W}\right)\right)}
    W^{(c-1)a}\right] 
    = (1+a)^d \\
\end{align}
The scaling properties of the coupled entropy are
\begin{align}
    \lim_{W\rightarrow\infty}\left[
    \frac{\ln_{\alpha\kappa} (\lambda W)
    ^{\frac{1}{1+d\kappa}}}
    { \ln_{\alpha\kappa} W^{\frac{1}{1+d\kappa}}}\right] 
    &= \lambda^\frac{\alpha\kappa/\gamma}{1+d\kappa} \\
    \lim_{W\rightarrow\infty}\left[
    \frac{\ln_{\alpha\kappa} W^{\frac{1+a}{1+d\kappa}}}
    {\ln_{\alpha\kappa} W^{\frac{1}{1+d\kappa}}}
    W^{-\frac{\alpha\kappa/\gamma}{1+d\kappa}a}\right] 
    &= 1
\end{align}
Therefore, $c=1-\frac{\alpha\kappa/\gamma}{1+d\kappa},$ which for $\gamma=1$ is $c=2-q$ and thus the same classification as the Tsallis $2-q$  and Rényi entropies. Likewise, $d=0,$ which is also the scaling of the Tsallis $2-q$ entropy.

There are additional scaling properties that could further characterize the coupled entropy.

\section{Mathematica Github Repository} \label{app_github}

A Github repository regarding Nonlinear Statistical Coupling methods is maintained at \cite{nelson_nonlinear_2025}. Within the repository is a folder "nsc-mathematica". The file "Coupled Entropy NSP 2025.nb" and its pdf copy "Coupled Entropy NSP 2025Nov20.pdf" contain the computations for the graphics in this paper. The file calls functions from "Coupled Functions.nb", which has a pdf copy "Coupled Functions 2025Nov20.pdf". Some additional computations are completed in "Compare Uncertainty Functions v4.nb", with pdf copy "Compare Uncertainty Functions v4 2025Nov20.pdf". 

\bibliography{MICS_BibTex}

\end{document}